\documentclass[preprint,12pt]{elsarticle}

\usepackage{amsmath}
\usepackage{amssymb}
\usepackage{amsthm}
\usepackage{mathpartir}
\usepackage{thmtools}
\usepackage{thm-restate}
\usepackage{stmaryrd}

\usepackage{microtype}

\usepackage{url}

\journal{J. Logical and Algebraic Methods in Programming }

\newcommand{\pipe}{\;\;|\;\;}

\newcommand{\pom}{\ensuremath{\mathsf{Pom}}}
\newcommand{\pomsp}{\ensuremath{\pom^\mathsf{sp}}}

\newcommand{\atrace}[2][]{\mathrel{\raisebox{-3pt}{$\xrightarrow[#1]{#2}$}}}
\newcommand{\tracerel}{\mathrel{\rightarrow}}

\newcommand{\ssderiv}{\delta_\Sigma}
\newcommand{\psderiv}{\gamma_\Sigma}
\newcommand{\satrace}[1]{\atrace{#1}_\Sigma}
\newcommand{\terms}{\mathcal{T}}
\newcommand{\sacc}{\mathcal{F}}
\newcommand{\naturals}{\mathbb{N}}

\newcommand{\rec}{\mathsf{rec}}
\newcommand{\seq}{\mathsf{seq}}

\usepackage{paralist}
\usepackage{subcaption}

\usepackage{tikz}
\usetikzlibrary{positioning}
\usetikzlibrary{automata}
\usetikzlibrary{angles}
\usetikzlibrary{arrows.meta}
\usetikzlibrary{decorations.pathmorphing}

\tikzset{ptransleft/.style={-{Latex[left,width=2mm]}}}
\tikzset{ptransright/.style={-{Latex[right,width=2mm]}}}
\tikzset{ptransleftinv/.style={{Latex[left,width=2mm]}-}}

\DeclareFontFamily{U}{matha}{}
\DeclareFontShape{U}{matha}{m}{n}{%
  <-5.5>    matha5
  <5.5-6.5> matha6
  <6.5-7.5> matha7
  <7.5-8.5> matha8
  <8.5-9.5> matha9
  <9.5-11>  matha10
  <11->     matha12
}{}
\DeclareSymbolFont{matha}{U}{matha}{m}{n}
\DeclareFontSubstitution{U}{matha}{m}{n}
\DeclareFontFamily{U}{mathx}{\hyphenchar\font45}
\DeclareFontShape{U}{mathx}{m}{n}{<-> mathx10}{}
\DeclareSymbolFont{mathx}{U}{mathx}{m}{n}
\DeclareFontSubstitution{U}{mathx}{m}{n}

\DeclareMathDelimiter{\ldbrack}{4}{matha}{"76}{mathx}{"30}
\DeclareMathDelimiter{\rdbrack}{5}{matha}{"77}{mathx}{"38}

\newcommand{\semnostretch}[1]{\ldbrack#1\rdbrack}
\newcommand{\sem}[1]{%
  \mathchoice%
  {\left\ldbrack#1\right\rdbrack} 
  {\semnostretch{#1}} 
  {\semnostretch{#1}} 
  {\semnostretch{#1}} 
}

\newcommand{\depth}[1]{\partial(#1)}
\newcommand{\pdepth}{d_\parallel}
\newcommand{\ddepth}{d_\dagger}

\newcommand{\tuple}[1]{\left\langle{#1}\right\rangle}
\newcommand{\set}[1]{\left\{{#1}\right\}}
\newcommand{\pomset}[1]{\langle\!|#1|\!\rangle}

\newtheorem{theorem}{Theorem}[section]
\newtheorem{lemma}[theorem]{Lemma}
\newtheorem{corollary}[theorem]{Corollary}
\theoremstyle{definition}
\newtheorem{definition}[theorem]{Definition}
\newtheorem{example}[theorem]{Example}

\usepackage{etoolbox}
\apptocmd{\sloppy}{\hbadness10000\relax}{}{}

\begin{document}

\begin{frontmatter}

\title{On Series-Parallel Pomset Languages: \\ Rationality, Context-Freeness and Automata\tnoteref{extendedversion}\tnoteref{profoundnet}}
\tnotetext[extendedversion]{This paper is an extended version of a paper published at CONCUR'17~\cite{kappe-brunet-luttik-silva-zanasi-2017}.}
\tnotetext[grants]{%
    This work was partially supported by the ERC Starting Grant ProFoundNet (grant code 679127).
    F.\ Zanasi acknowledges support from EPSRC grant n.\ EP/R020604/1.
}
\author[ucl]{Tobias~Kapp\'{e}}
\ead{tkappe@cs.ucl.ac.uk}
\author[ucl]{Paul~Brunet}
\author[tue]{Bas~Luttik}
\author[ucl]{Alexandra~Silva}
\author[ucl]{Fabio~Zanasi}
\address[ucl]{University College London, London, United Kingdom}
\address[tue]{Eindhoven University of Technology, Eindhoven, The Netherlands}

\begin{abstract}
Concurrent Kleene Algebra (CKA) is a formalism to study concurrent programs.
Like previous Kleene Algebra extensions, developing a correspondence between denotational and operational perspectives is important, for both foundations and applications.
This paper takes an important step towards such a correspondence, by precisely relating bi-Kleene Algebra (BKA), a fragment of CKA, to a novel type of automata, \emph{pomset automata} (PAs).

We show that PAs can implement the BKA semantics of series-parallel rational expressions, and that a class of PAs can be translated back to these expressions.
We also characterise the behavior of general PAs in terms of context-free pomset grammars; consequently, universality, equivalence and series-parallel rationality of general PAs are undecidable.
\end{abstract}


\end{frontmatter}

\section{Introduction}

In their CONCUR'09 paper~\cite{hoare-moeller-struth-wehrman-2009}, Hoare, M\"{o}ller, Struth, and Wehrman introduced \emph{Concurrent Kleene Algebra} (CKA) as a mathematical framework suitable for the study of concurrent programs, in the hope of achieving the same elegance that Kozen did when using Kleene Algebra (KA) and extensions to provide a verification platform for sequential programs.

CKA is a seemingly simple extension of KA\@: it adds parallel analogues to the sequential composition and Kleene star operators, as well as the \emph{exchange law}, which axiomatises interleaving.
Extending the KA toolkit, however, is challenging; in particular, an operational perspective is missing.
In contrast, the correspondence between denotational and operational aspects of KA is well-understood through Kleene's theorem~\cite{kleene-1956}, which provided a pillar for characterising the free model~\cite{kozen-1997} and establishing a decision procedure~\cite{hopcroft-karp-1971}.

With this in mind, we pursue a version of Kleene's theorem for CKA\@.
Specifically, we study \emph{series-parallel rational expressions} (\emph{spr-expressions}), with a denotational model in terms of pomset languages.
Our main contribution is a theorem which faithfully relates these expressions to a newly defined automaton model, called \emph{pomset automata} (\emph{PAs}).
In a nutshell, PAs are automata where traces from certain states may branch into parallel threads; these threads contribute to the language when both reach an accepting state.

We are not the first to attempt such a characterisation.
However, earlier works~\cite{lodaya-weil-2000,jipsen-moshier-2016} fall short of giving a precise correspondence between the denotational and operational models, due to the lack of a structural restriction on automata ensuring that only valid behaviours are accepted.
In contrast, we propose such a restriction, which guarantees the soundness of a translation from the operational to the denotational model.
Furthermore, we propose a generalisation of Brzozowski derivatives~\cite{brzozowski-1964} in the translation from expressions to automata, avoiding unnecessary $\epsilon$-transitions and non-determinism that would result from a construction in the style of Thompson~\cite{thompson-1968}.

Since our denotational model does not take interleaving into account (and hence is not sound for the exchange law), our work is most accurately described as an operational model for \emph{bi-Kleene Algebra} (BKA)~\cite{laurence-struth-2014}.
We leave it to future work to incorporate the exchange law.

This work extends the conference paper~\cite{kappe-brunet-luttik-silva-zanasi-2017} published at CONCUR'17 with previously omitted proofs and two new results.
The first is the extension of the main theorem to incorporate the parallel variant of the Kleene star operator.
The second is a characterisation of the behaviors of finite pomset automata in terms of context-free grammars (CFGs)~\cite{chomsky-1956}.

\medskip
The paper is organised as follows.
We recall preliminaries in Section~\ref{section:preliminaries}, and introduce PAs in Section~\ref{section:pomset-automata}.
We translate a class of PAs to equivalent spr-expressions in Section~\ref{section:automata-to-expressions}, and describe the reverse construction in Section~\ref{section:expressions-to-automata}.
We characterise finite PAs in terms of CFGs in Section~\ref{sec:cfg}.
We discuss related work in Section~\ref{section:related-work}; directions for further work appear in Section~\ref{section:further-work}.

To preserve the flow of the narrative, some proofs appear in the appendices; routine proofs are omitted altogether.

\section{Preliminaries}%
\label{section:preliminaries}

Let $S$ be a set; we write $2^S$ for the set of all subsets of $S$.
We refer to a relation $\prec$ on $S$ as \emph{well-founded} if there are no infinite descending $\prec$-chains, i.e., no $\set{s_n}_{n\in\naturals} \subseteq S$ such that for all $n \in \naturals$ it holds that $s_{n+1} \prec s_n$.

Throughout the paper we fix a finite set $\Sigma$ called the \emph{alphabet}, whose elements are symbols usually denoted by $a$, $b$, etc.
Lastly, if ${\rightarrow} \subseteq X \times Y \times Z$ is a ternary relation, we write $x \xrightarrow{y} z$ instead of $\tuple{x, y, z} \in {\rightarrow}$.

\subsection{Pomsets}

\emph{Partially-ordered multisets}, or \emph{pomsets}~\cite{gischer-1988,grabowski-1981} for short, generalise words to a setting where actions (elements from $\Sigma$) may take place not just sequentially, but also in parallel.
We recall a rigorous definition of pomsets, as well as some useful fundamental notions from literature~\cite{gischer-1988,grabowski-1981,lodaya-weil-2000,esik-nemeth-2004,laurence-struth-2014}.

\begin{definition}%
\label{definition:pomset}
A \emph{labelled poset} is a tuple $\tuple{C, \leq_C, \lambda_C}$ consisting of a \emph{carrier} set $C$, a partial order $\leq$ on $C$ and a \emph{labelling} $\lambda: C \to \Sigma$.

A \emph{labelled poset isomorphism} is a bijection between carriers that bijectively preserves labels and ordering.
A \emph{pomset} is an isomorphism class of labelled posets; we use $\pomset{C, \leq, \lambda}$ to denote the pomset represented by $\tuple{C, \leq, \lambda}$.
\end{definition}

For instance, suppose a recipe for caramel-glazed cookies tells us to
\begin{inparaenum}[(i)]
    \item\label{recipe:dough} \emph{prepare} cookie dough,
    \item\label{recipe:bake} \emph{bake} cookies in the oven,
    \item\label{recipe:caramelise} \emph{caramelise} sugar,
    \item\label{recipe:glaze} \emph{glaze} the finished cookies.
\end{inparaenum}
Here, step~\eqref{recipe:dough} precedes steps~\eqref{recipe:bake} and~\eqref{recipe:caramelise}.
Furthermore, step~\eqref{recipe:glaze} succeeds both steps~\eqref{recipe:bake} and~\eqref{recipe:caramelise}.
A pomset representing this process could be $U = \pomset{C_U, \leq_U, \lambda_U}$, where $C_U = \set{\eqref{recipe:dough},\eqref{recipe:bake},\eqref{recipe:caramelise},\eqref{recipe:glaze}}$ and $\leq_U$ is such that $\eqref{recipe:dough}\leq_U\eqref{recipe:bake}\leq_U\eqref{recipe:glaze}$ and $\eqref{recipe:dough}\leq_U\eqref{recipe:caramelise}\leq_U\eqref{recipe:glaze}$; $\lambda_U$ associates with the elements of $C_U$ the corresponding steps in the recipe.

We use $\pom$ to denote the collection of all pomsets.
Labelled posets and pomsets with a countable carrier suffice for our purposes.
For this reason, we can (w.l.o.g.) adopt the convention that the carrier of a labelled poset representing a pomset is a subset of $\naturals$, which makes $\pom$ a proper set.

Words over $\Sigma$ are identified with totally ordered pomsets; multisets over $\Sigma$ are similarly identified with pomsets having a discrete (diagonal) order.
We write $1$ for the empty pomset, and use $a \in \Sigma$ to refer to the \emph{primitive} pomset with a single point labelled $a$ (and the obvious order).
Finally, we use the symbols $U, V, \ldots$ to denote pomsets.

\begin{definition}%
\label{definition:pomset-composition}
Let $U = \pomset{C_U, \leq_U, \lambda_U}$ and $V = \pomset{C_V, \leq_V, \lambda_V}$ be pomsets.
Without loss of generality, we can assume that $C_U$ and $C_V$ are disjoint.

The \emph{sequential composition} of $U$ and $V$, denoted $U \cdot V$, is the pomset
\[\pomset{C_U \cup C_V,\;{\leq_U} \cup {\leq_V} \cup (C_U \times C_V),\;\lambda_U \cup \lambda_V}\]
The \emph{parallel composition} of $U$ and $V$, denoted $U \parallel V$, is the pomset
\[\pomset{C_U \cup C_V,\;{\leq_U} \cup {\leq_V},\;\lambda_U \cup \lambda_V}\]
Here, $\lambda_U \cup \lambda_V: C_U \cup C_V \to \Sigma$ agrees with $\lambda_U$ on $C_U$, and with $\lambda_V$ on $C_V$.
\end{definition}

As a convention, sequential composition takes precedence over parallel composition, i.e., $U \cdot V \parallel W$ is read as $(U \cdot V) \parallel W$.

Sequential composition forces the events in the left pomset to be ordered before those in the right pomset.
We note that these operators are well-defined modulo isomorphism of labelled posets, and that the empty pomset $1$ is the unit for both sequential and parallel composition.

\begin{definition}%
\label{definition:pomset-series-parallel}
The set of \emph{series-parallel} pomsets~\cite{gischer-1988,grabowski-1981}, $\pomsp$, is the smallest subset of $\pom$ that includes the empty and primitive pomsets and is closed under sequential and parallel composition.
\end{definition}

In this paper we concern ourselves with series-parallel pomsets.
For inductive reasoning, it is useful to recall part of~\cite[Theorem 3.1]{gischer-1988}.
\begin{lemma}%
\label{lemma:pomset-unique-factorisation}
Let $U \in \pomsp$.
If $U$ is non-empty, then \emph{exactly one} of the following is true:
\begin{inparaenum}[(i)]
    \item\label{case:decompose-primitive} $U = a$ for some $a \in \Sigma$, or
    \item\label{case:decompose-sequential} $U = V \cdot W$ for non-empty $V, W \in \pomsp$, strictly smaller than $U$, or
    \item\label{case:decompose-parallel} $U = V \parallel W$ for non-empty $V, W \in \pomsp$, strictly smaller than $U$.
\end{inparaenum}
\end{lemma}

We can quantify the degree of nesting of parallel and sequential composition of a series-parallel pomset as follows.

\begin{definition}%
\label{definition:pomset-depth}
The \emph{depth} of a series-parallel pomset $U$~\cite{esik-nemeth-2004}, denoted $\depth{U}$, is defined inductively, as follows.
First, if $U = 1$, then $\depth{U} = 0$.
Second, if $U = a$ for some $a \in \Sigma$, then $\depth{U} = 1$.
Third, if $U = U_0 \cdots U_{n-1}$ or $U = U_0 \parallel \dots \parallel U_{n-1}$ for non-empty pomsets $U_0, \dots, U_{n-1}$, and $n > 1$ is maximal for such a decomposition, then
\[\depth{U} = \max(\depth{U_0}, \dots, \depth{U_{n-1}}) + 1\]
\end{definition}

Note that depth is always well-defined, as a consequence of Lemma~\ref{lemma:pomset-unique-factorisation}.

\subsection{Pomset languages}

We can group the words that represent traces arising from a sequential program into a set called a \emph{language}.
By analogy, we can group the pomsets that represent the traces arising from a parallel program into a \emph{pomset language}.
We use calligraphic symbols $\mathcal{U}, \mathcal{V}, \ldots$ to denote pomset languages.

For instance, suppose that the recipe for glazed cookies has an optional fifth step where chocolate sprinkles are spread over the cookies.
In that case, there are \emph{two} pomsets that describe a trace arising from the recipe, $U^+$ and $U^-$, either with or without the chocolate sprinkles.
The pomset language $\mathcal{U} = \set{U^-, U^+}$ contains the traces that arise from the new recipe.

The composition operators for pomsets can be lifted to pomset languages.
There also exist two types of Kleene closure operator, similar to the one defined on languages of words, for both parallel and sequential composition.
\begin{definition}%
\label{definition:pomset-language-composition}
Let $\mathcal{U}$ and $\mathcal{V}$ be pomset languages.
We define:
\begin{align*}
\mathcal{U} \cdot \mathcal{V} &= \set{U \cdot V : U \in \mathcal{U}, V \in \mathcal{V}} &
    \mathcal{U}^* &= \bigcup_{n \in \naturals} \mathcal{U}^n \\
\mathcal{U} \parallel \mathcal{V} &= \set{U \parallel V : U \in \mathcal{U}, V \in \mathcal{V}} &
    \mathcal{U}^\dagger &= \bigcup_{n \in \naturals} \mathcal{U}^{(n)}
\end{align*}
in which $\mathcal{U}^0 = \mathcal{U}^{(0)} = \set{1}$, and for all $n \in \naturals$ we define
\begin{mathpar}
\mathcal{U}^{n+1} = \mathcal{U} \cdot \mathcal{U}^n
\and
\mathcal{U}^{(n+1)} = \mathcal{U} \parallel \mathcal{U}^{(n)}
\end{mathpar}
\end{definition}
Sequential Kleene closure models indefinite repetition.
For instance, if our cookie recipe has a final step ``repeat as necessary'', the pomset language $\mathcal{U}^*$ represents all possible traces of repetitions of the recipe; e.g., $U^+ \cdot U^+ \cdot U^- \in \mathcal{U}^*$ is the trace where first two batches of sprinkled cookies are made, followed by one without sprinkles.
In contrast, parallel Kleene closure models unbounded parallelism; in this case, $\mathcal{U}^\dagger$ represents all possible traces of parallel executions of the recipe; e.g., $U^+ \parallel U^- \in \mathcal{U}^\dagger$ is the trace where we make two batches of cookies in parallel, one with and one without sprinkles.

\subsection{Series-parallel rational expressions}

Just as a rational expression can be used to describe a rational structure of sequential events, so too can a series-parallel rational expression be used to describe a rational structure of possibly parallel events.
Series-parallel rational expressions can be thought of as rational expressions with parallel composition, as well as a parallel analogue to the Kleene star.

\begin{definition}%
\label{definition:series-parallel-rational-expression}
We use $\terms$ to denote the set of \emph{series-parallel rational expressions} (\emph{spr-expressions}, for short)~\cite{lodaya-weil-2000}, formed by the grammar
\[e, f ::= 0 \pipe 1 \pipe a \in \Sigma \pipe e + f \pipe e \cdot f \pipe e \parallel f \pipe e^* \pipe e^\dagger\]
For a closed propositional formula $\Phi$, we write $[\Phi]$ as shorthand for $1$ if $\Phi$ is satisfied, and $0$ otherwise.
We use $e$, $f$, $g$ and $h$ to denote spr-expressions.
\end{definition}

Series-rational expressions have a semantics in terms of pomset languages.

\begin{definition}%
\label{definition:series-parallel-rational-expression-semantics}
The function $\sem{-}: \terms \to 2^{\pomsp}$ is defined~\cite{lodaya-weil-2000} inductively:
\begin{align*}
\sem{0} &= \emptyset & \sem{a} &= \set{a} & \sem{e \cdot f} &= \sem{e} \cdot \sem{f} & \semnostretch{e^*} &= \sem{e}^* \\
\sem{1} &= \set{1} & \sem{e + f} &= \sem{e} \cup \sem{f} &
\sem{e \parallel f} &= \sem{e} \parallel \sem{f} &
\semnostretch{e^\dagger} &= \sem{e}^\dagger
\end{align*}
If $\mathcal{U} \subseteq \pomsp$ such that $\mathcal{U} = \sem{e}$ for some $e \in \terms$, then $\mathcal{U}$ is said to be a \emph{series-parallel rational language}, or \emph{spr-language} for short.
\end{definition}

To illustrate, recall the pomset language $\mathcal{U}^* = \set{U^+, U^-}^*$.
We can describe $\set{U^-}$ and $\set{U^+}$ with the series-parallel rational expressions
\begin{mathpar}
e^- = \mathsf{prepare} \cdot (\mathsf{bake} \parallel \mathsf{caramelise}) \cdot \mathsf{glaze}
\and
e^+ = e^- \cdot \mathsf{sprinkle}
\end{mathpar}
which yields the spr-expression $e = e^- + e^+$ for $\mathcal{U}$; hence, $\sem{e^*} = \mathcal{U}^*$.

Note that spr-expressions without $\parallel$ and $\dagger$ are rational expressions, and spr-expressions without $\cdot$ and $*$ are commutative rational expressions~\cite{conway-1971}.
To see that spr-expressions are a proper extension of rational and commutative rational expressions, we observe the following.
\begin{lemma}%
\label{lemma:proper-extension}
Let $e \in \terms$.
The following are true.
\begin{inparaenum}[(i)]
    \item
    If $\sem{e}$ consists of words, then $\sem{e}$ is a rational language.
    \item
    If $\sem{e}$ consists of multisets, then $\sem{e}$ is a commutative rational language.
\end{inparaenum}
\end{lemma}

We conclude our discussion of pomset languages by recalling the following lemma, which is useful when analysing the series-parallel rationality of a language.
For details, refer to~\cite{lodaya-weil-2000,laurence-struth-2014}.
\begin{lemma}%
\label{lemma:spr-languages-bounds}
If $\mathcal{U}$ is an spr-language, then there exists an $n \in \naturals$ such that for all $U \in \mathcal{U}$ it holds that $\depth{U} \leq n$.
\end{lemma}

More specifically, the above lemma tells us that when we want to show that a pomset language $\mathcal{U}$ is not series-parallel rational, it suffices to find a sequence $\set{U_n}_{n \in \naturals} \subseteq \mathcal{U}$ such that for $n \in \naturals$ we have $\depth{U_n} < \depth{U_{n+1}}$.

\section{Pomset Automata}%
\label{section:pomset-automata}

We now describe an automaton model to recognise pomset languages.

\begin{definition}%
\label{definition:automaton}
A \emph{pomset automaton (PA)} is a tuple $\tuple{Q, \delta, \gamma, F}$ where
    $Q$ is a set of \emph{states}, with $F \subseteq Q$ the \emph{accepting states};
    $\delta: Q \times \Sigma \to Q$ is a function called the \emph{sequential transition function}, and
    $\gamma: Q \times Q \times Q \to Q$ is a function called the \emph{parallel transition function}.
\end{definition}

We do not fix an initial state; thus, a PA does not define a single pomset language but rather a mapping from states to pomset languages.
This mapping is defined in terms of a trace relation arising from $\delta$ and $\gamma$, as follows.

\begin{definition}%
\label{definition:automaton-language}
Let $A = \tuple{Q, \delta, \gamma, F}$ be a PA\@.
We define ${\tracerel_A} \subseteq Q \times \pomsp \times Q$ as the smallest relation that satisfies the rules
\begin{mathpar}
\inferrule{%
    q \in Q
}{%
    q \atrace{1}_A q
}
\and
\inferrule{%
    q \in Q \\
    a \in \Sigma
}{%
    q \atrace{a}_A \delta(q, a)
}
\and
\inferrule{%
    q \atrace{U}_A q'' \\\\
    q'' \atrace{V}_A q'
}{%
    q \atrace{U \cdot V}_A q'
}
\and
\inferrule{%
    r \atrace{U}_A r' \in F \\\\
    s \atrace{V}_A s' \in F
}{%
    q \atrace{U \parallel V}_A \gamma(q, r, s)
}
\end{mathpar}
The \emph{language} of $A$ at $q \in Q$ is $L_A(q) = \set{U \in \pomsp : \exists q' \in F.\ q \atrace{U}_A q'}$.
We say that $A$ \emph{accepts} the pomset language $\mathcal{U}$ if $L_A(q) = \mathcal{U}$ for some $q \in Q$.
\end{definition}

In the above, $\delta$ plays the same role as in classic finite automata: given a state and a symbol, it returns the new state after reading that symbol.
The function $\gamma$ deserves a bit more explanation: given states $q, r, s \in Q$, it tells us the state that is reached from $q$ after reading two pomsets in parallel starting at states $r$ and $s$, and having both reach an accepting state.

For the remainder of this section, we fix a PA $A = \tuple{Q, \delta, \gamma, F}$.
Individual triplets in the trace relation are referred to as \emph{traces}.
It is useful to establish some terminology when referring to traces.
Specifically, for all $q \in Q$:
\begin{itemize}
    \setlength{\itemsep}{0em}
    \item
    We define $q \atrace{1}_A q$ as a \emph{trivial} trace.

    \item
    For all $a \in \Sigma $, we define $q \atrace{a}_A \delta (q, a)$ as a \emph{$\delta $-trace}.

    \item
    For all traces $r \atrace{U}_A r'$ and $s \atrace{V}_A s'$ with $r', s' \in F$, we define $q \atrace{U \parallel V}_A \gamma(q, r, s)$ as a \emph{$\gamma$-trace}.
\end{itemize}
The $\delta$-traces and $\gamma$-traces are collectively known as \emph{unit} traces.

\medskip
To simplify matters later on, we assume that every PA $A$ has states $\bot \in Q \setminus F$ and $\top \in F$ such that
\begin{inparaenum}[(i)]
    \item for all $a \in \Sigma$, it holds that $\delta(\bot, a) = \delta(\top, a) = \bot$, and
    \item for all $r, s \in Q$, it holds that $\gamma(\bot, r, s) = \gamma(\top, r, s) = \bot$.
\end{inparaenum}

The state $\bot$ is useful when defining $\gamma$: for a fixed $q \in Q$, not all $r, s \in Q$ may give a value of $\gamma(q, r, s)$ that contributes to $L_A(q)$; in such cases, we set $\gamma(q, r, s) = \bot$.%
\footnote{Alternatively, we could have allowed $\gamma$ to be a partial function; including $\bot$ as a state, however, will simplify part of our construction in Section~\ref{section:expressions-to-automata}.}
The state $\top$ fulfills a similar role: it is used to signal that the target of a parallel transition accepts, but allows no further continuation of the trace; this will be important in Section~\ref{section:automata-to-expressions}, when we describe a class of pomset automata that admit a translation back to spr-expressions.

\begin{restatable}{lemma}{tracetopbottom}%
\label{lemma:trace-top-bottom}
Let $q \atrace{U}_A q'$ be non-trivial.
If $q = \bot$ or $q = \top$, then $q' = \bot$.
\end{restatable}

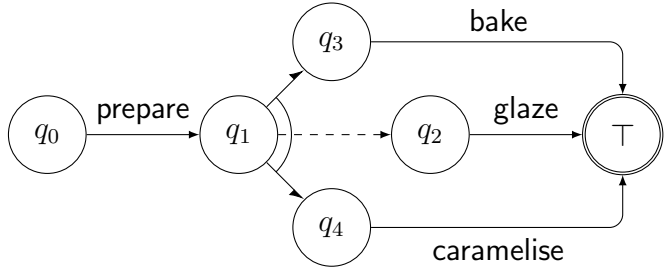
\begin{figure}
    \centering
    \begin{tikzpicture}[every node/.style={transform shape}]
        \node[state] (q0) {$q_0$};
        \node[state,right=15mm of q0] (q1) {$q_1$};
        \node[state,above right=7mm of q1] (q3) {$q_3$};
        \node[state,below right=7mm of q1] (q4) {$q_4$};
        \node[state,right=15mm of q1] (q2) {$q_2$};
        \node[state,accepting,right=15mm of q2] (top) {$\top$};

        \draw[-latex] (q0) edge node[above] {\textsf{prepare}} (q1);
        \draw[ptransleft] (q1) edge (q3);
        \draw[ptransright] (q1) edge (q4);
        \draw[dashed,-latex] (q1) edge (q2);
        \draw[-latex, rounded corners=5pt] (q3) -| node[above,xshift=-4em] {\textsf{bake}} (top);
        \draw[-latex, rounded corners=5pt] (q4) -| node[below,xshift=-4em] {\textsf{caramelise}} (top);
        \draw pic[draw,angle radius=7mm] {angle=q4--q1--q3};
        \draw[-latex] (q2) edge node[above] {\textsf{glaze}} (top);
    \end{tikzpicture}
    \caption{Pomset automaton accepting $U$.}\label{figure:example-pa}
\end{figure}

We draw a PA in a way similar to finite automata: each state (except $\bot$) is a vertex, and accepting states are marked by a double border.
To represent sequential transitions, we draw labelled edges; for instance, in Figure~\ref{figure:example-pa}, $\delta(q_0, \mathsf{prepare}) = q_1$.
To represent parallel transitions, we draw hyper-edges; for instance, in Figure~\ref{figure:example-pa}, $\gamma(q_1, q_3, q_4) = q_2$.
To avoid clutter, we do not draw either of these edge types when the target state is $\bot$.
It is not hard to verify that the pomset $U$ of the earlier example is accepted by the PA in Figure~\ref{figure:example-pa}.

\subsection{Finite support}

Deterministic automata with infinitely many states can accept non-rational languages.
Since spr-languages extend rational languages by Lemma~\ref{lemma:proper-extension}, and PAs obviously extend deterministic automata, it follows that allowing PAs with infinitely many states would dash our hopes of a Kleene theorem.

On the other hand, it is useful to work with PAs that have infinitely many states, as we shall see in Section~\ref{section:expressions-to-automata}.
To strike a middle ground, we identify a class of PAs with possibly infinitely many states that, for any state $q$, allow a restriction to a PA with finitely many states accepting the language of $q$.

\begin{definition}
The \emph{trace dependency relation} of $A$, denoted $\preceq_A$, is the smallest preorder on $Q$ that satisfies the rules
\begin{mathpar}
\inferrule{%
    q, r, s \in Q \\
    \gamma(q, r, s) \neq \bot
}{%
    r, s \preceq_A q
}
\and
\inferrule{%
    a \in \Sigma \\
    q \in Q
}{%
    \delta(q, a) \preceq_A q
}
\and
\inferrule{%
    q, r, s \in Q
}{%
    \gamma(q, r, s) \preceq_A q
}
\end{mathpar}
It should be emphasised that, in general, $\preceq_A$ is not a partial order --- antisymmetry may fail because of loops in the transition structure.

We write $\prec_A$ for the \emph{strict trace dependency relation}, which is the strict order that arises by setting $q \prec_A q'$ if and only if $q \preceq_A q'$ and $q' \not\preceq_A q$.
\end{definition}

\begin{definition}
We say that $Q' \subseteq Q$ is \emph{closed} in $A$ when $Q'$ is downward-closed with respect to $\preceq_A$ --- that is, for all $q \in Q'$ and $r \in Q$ such that $r \preceq_A q$, it follows that $r \in Q'$.
We write $\pi_A(q)$ for the \emph{support} of $q$ in $A$, which is the smallest closed subset of $Q$ that contains $q$.
We say that $A$ is \emph{finitely supported} if for all $q \in Q$ it holds that $\pi_A(q)$ is finite.
\end{definition}

With this definition, the following is not hard to see.

\begin{restatable}{lemma}{restrictfinitelysupportedautomaton}%
\label{lemma:restrict-finitely-supported-automaton}
If $A$ is finitely supported, then for every $q \in Q$ there exists a finite pomset automaton $A_q$ with a state $q'$, such that $L_A(q) = L_{A_q}(q')$.
\end{restatable}

Finite support is also useful in that it ensures well-foundedness of the strict trace dependency relation.
\begin{restatable}{lemma}{finitelysupportedimplieswellfounded}
If $A$ is finitely supported, then $\prec_A$ is well-founded.
\end{restatable}

\subsection{Trace length}

We conclude this section with the following technical lemma, which gives us an alternative inductive handle for the lemmas to come.

\begin{restatable}{lemma}{tracelength}%
\label{lemma:trace-length}
If $q \atrace{U}_A q'$, then there exist $q_0, \dots, q_\ell \in Q$ with $q = q_0$ and $q_\ell = q'$, and $U = U_0 \cdots U_{\ell-1}$ such that for $0 \leq i < \ell$ it holds that $q_i \atrace{U_i}_A q_{i+1}$.
Furthermore, each of these traces is a unit trace.
\end{restatable}

The minimal $\ell$ for a given trace as obtained from the above lemma is known as the \emph{length} of the trace.
Note that a trace of zero length is necessarily trivial, and a trace of unit length is necessarily a unit trace.

\section{Automata to expressions}%
\label{section:automata-to-expressions}

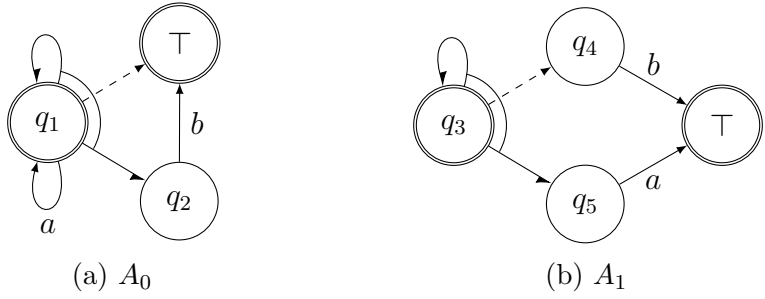
\begin{figure}
    \centering
    \begin{subfigure}[b]{0.4\textwidth}
        \centering
        \begin{tikzpicture}[every loop/.style={}]
            \node[state,accepting] (q1) {$q_1$};
            \node[state,below right=3mm and 10mm of q1] (q2) {$q_2$};
            \node[state,accepting,above right=3mm and 10mm of q1] (top) {$\top$};

            \draw[ptransleftinv] (q1) edge[loop above] (q1);
            \draw[-latex] (q1) edge[loop below] node {$a$} (q1);
            \draw[ptransright] (q1) edge (q2);
            \draw[dashed,-latex] (q1) edge (top);
            \draw[-latex] (q2) edge node[right] {$b$} (top);
            \draw (q1) + (-30:7mm) arc (-30:76:7mm);
        \end{tikzpicture}
        \caption{$A_0$}\label{figure:problem-self-loop-sequential-loop}
    \end{subfigure}
    \begin{subfigure}[b]{0.5\textwidth}
        \centering
        \begin{tikzpicture}[every loop/.style={}]
            \node[state,accepting] (q3) {$q_3$};
            \node[state,above right=3mm and 10mm of q3] (q4) {$q_4$};
            \node[state,accepting,right=25mm of q3] (q6) {$\top$};
            \node[state,below right=3mm and 10mm of q3] (q5) {$q_5$};

            \draw[-latex] (q4) edge node[above] {$b$} (q6);
            \draw[-latex] (q5) edge node[below] {$a$} (q6);

            \draw[ptransleftinv] (q3) edge[loop above] (q3);
            \draw[ptransright] (q3) edge (q5);
            \draw[dashed,-latex] (q3) edge (q4);
            \draw (q3) + (-30:7mm) arc (-30:76:7mm);
        \end{tikzpicture}
        \caption{$A_1$}\label{figure:problem-self-loop-diversify}
    \end{subfigure}
    \caption{Finitely supported pomset automata that accept languages of unbounded depth.}\label{figure:problem-self-loop}
\end{figure}

Let us fix a finitely supported PA $A = \tuple{Q, \delta, \gamma, F}$.
We set out to find for every $q \in Q$ an $e_q \in \terms$ such that $L_A(q) = \sem{e_q}$.
Before we get started, however, it should be noted that not all finitely supported (or even finite) pomset automata admit such a translation.
This is because $\delta$ and $\gamma$ can conspire to create a state with a language of unbounded depth; Lemma~\ref{lemma:spr-languages-bounds} then tells us that the corresponding spr-expression cannot exist.

\begin{example}
Consider the PA $A_0$ in Figure~\ref{figure:problem-self-loop-sequential-loop}.
Here, we have that
\begin{mathpar}
q_1 \atrace{a}_{A_0} \delta(q_1, a) = q_1
\and
q_2 \atrace{b}_{A_0} \delta(q_2, b) = \top
\end{mathpar}
Since $q_1, \top \in F$, we find $q_1 \atrace{a \parallel b}_{A_0} \gamma(q_1, q_1, q_2) = \top$ and $q_1 \atrace{a \cdot (a \parallel b)}_{A_0} \top$.
Hence,
\begin{mathpar}
q_1 \atrace{a \cdot (a \parallel b) \parallel b}_{A_0} \gamma(q_1, q_1, q_2) = \top
\and
q_1 \atrace{a \cdot (a \cdot (a \parallel b) \parallel b)}_{A_0} \top
\end{mathpar}
We can repeat this indefinitely, thereby showing that
\[
    \set{1, a, a \parallel b, a \cdot (a \parallel b), a \cdot (a \parallel b) \parallel b, \dots} \subseteq L_{A_0}(q_1)
\]
Consequently, there is no $e \in \terms$ such that $L_{A_0}(q_1) = \sem{e}$, by Lemma~\ref{lemma:spr-languages-bounds}.
\end{example}

\begin{example}
Consider the PA $A_1$ in Figure~\ref{figure:problem-self-loop-diversify}.
Here, we have that
\begin{mathpar}
q_5 \atrace{a}_{A_1} \delta(q_5, a) = \top
\and
q_4 \atrace{b}_{A_1} \delta(q_4, b) = \top
\and
q_3 \atrace{1}_{A_1} q_3
\end{mathpar}
Since $q_3 \in F$, we find $q_3 \atrace{1 \parallel a}_{A_1} \gamma(q_3, q_3, q_5) = q_4$ and $q_3 \atrace{a \cdot b}_{A_1} \top$.
Hence,
\begin{mathpar}
q_3 \atrace{a \cdot b \parallel a}_{A_1} \gamma(q_3, q_3, q_5) = q_4
\and
q_3 \atrace{(a \cdot b \parallel a) \cdot b}_{A_1} \top
\end{mathpar}
We can repeat the above to show that $\set{1, a \cdot b, ((a \cdot b) \parallel a) \cdot b, \dots} \subseteq L_{A_1}(q_3)$.
By Lemma~\ref{lemma:spr-languages-bounds}, we then find that there is no $e \in \terms$ with $L_{A_1}(q_3) = \sem{e}$.
\end{example}

To get around the problem of unbounded languages, we structurally restrict pomset automata in such a way that such behavior is excluded.
To do this, we need to get a handle on the constellation of states and transitions common to the examples above that allows the depth of the pomset languages accepted by $A_0$ and $A_1$ to run amok; this is done in the following lemma.

\begin{lemma}
Let $q_0, q_2, q_4 \in Q$ and $q_1, q_3, q_5 \in F$.
Let $U, V, W, X \in \pomsp$ be such that the following (c.f.\ Figure~\ref{figure:problematic-pattern}) hold:
\begin{mathpar}
q_0 \atrace{U}_A q_1
\and
q_0 \atrace{V}_A q_0
\and
q_2 \atrace{X}_A q_3
\and
q_4 \atrace{W}_A q_5
\and
\gamma(q_0, q_2, q_0) = q_4
\end{mathpar}
If $X \neq 1$, and moreover $W \neq 1$ or $V \neq 1$, then $L_A(q_0)$ has unbounded depth.
\end{lemma}
\begin{proof}
Suppose that $Y \in L_A(q_0)$, i.e., $q_0 \atrace{Y}_A q'$ for some $q' \in F$.
Given the traces and the fork transition in the premises, we can then derive that
\[
    q_0 \atrace{(X \parallel (V \cdot Y)) \cdot W} q_5 \in F
\]
and hence $(X \parallel (V \cdot Y)) \cdot W \in L_A(q_0)$.
Thus, $L_A(q_0)$ is closed under
\[
    f: \pomsp \to \pomsp \quad \mbox{given by} \quad f({-}) = (X \parallel (V \cdot {-})) \cdot W
\]
By the premise that $X \neq 1$ as well as $W \neq 1$ or $V \neq 1$, it follows that for $Y \in \pomsp$ we have $\depth{Y} < \depth{f(Y)}$.
Because $U \in L_A(q_0)$, we can point to $\set{U, f(U), f^2(U), \dots}$ as a set of unbounded depth contained in $L_A(q_0)$, and thus conclude that this pomset language has unbounded depth.
\end{proof}

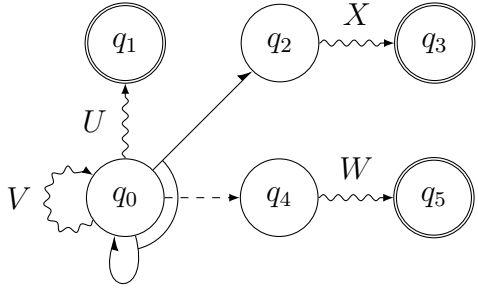
\begin{figure}
    \centering
    \begin{tikzpicture}[every loop/.style={}]
        \node[state] (q0) {$q_0$};
        \node[state,accepting,above=of q0] (q1) {$q_1$};
        \node[state,right=of q1] (q2) {$q_2$};
        \node[state,accepting,right=of q2] (q3) {$q_3$};
        \node[state,right=of q0] (q4) {$q_4$};
        \node[state,accepting,right=of q4] (q5) {$q_5$};

        \tikzset{decoration={snake,amplitude=.4mm,segment length=2mm,post length=1mm}}

        \draw (q0) edge[-latex,decorate] node[left=1mm] {$U$} (q1);
        \draw (q2) edge[-latex,decorate] node[above=1mm] {$X$} (q3);
        \draw (q4) edge[-latex,decorate] node[above=1mm] {$W$} (q5);
        \draw (q0) edge[loop left,decorate,-latex,out=-145,in=145,looseness=4] node[left=1mm] {$V$} (q0);

        \draw[ptransright] (q0) edge[loop below] (q0);
        \draw[ptransleft] (q0) edge (q2);
        \draw (q0) + (-76:7mm) arc (-76:45:7mm);
        \draw[-latex,dashed] (q0) edge (q4);
    \end{tikzpicture}
    \caption{A template for a state with a language of unbounded depth.}\label{figure:problematic-pattern}
\end{figure}

To counteract the pattern summarised above, we propose the following.

\begin{definition}%
\label{definition:well-nested}
We say that $q \in Q$ is \emph{sequential} if for all $r, s \in Q$ with $\gamma(q, r, s) \neq \bot$, it holds that $r, s \prec_A q$.
We say that $q \in F$ is \emph{recursive} if
\begin{inparaenum}[(i)]
    \item it is not sequential, and
    \item for all $a \in \Sigma$ we have $\delta(q, a) = \bot$, and
    \item if $r, s \in Q$ and $\gamma(q, r, s) \neq \bot$, then $s = q$ and $r \prec_A q$, and $\gamma(q, r, s) = \top$.
\end{inparaenum}

We write $Q_\seq$ (resp.\ $Q_\rec$) for the set of states in $Q$ that are sequential (resp.\ recursive), and say that $A$ is \emph{well-nested} if $Q = Q_\seq \cup Q_\rec$.
\end{definition}

One easily sees that $A_0$ and $A_1$ are not well-nested: neither $q_1$ nor $q_3$ is sequential, because of their self-forks, but $q_1$ is not recursive because $\delta(q_1, a) = q_1 \neq \bot$, and $q_3$ is not recursive because $\gamma(q_3, q_5, q_3) = q_4 \not\in \set{\top, \bot}$.

As a matter of fact, Definition~\ref{definition:well-nested} is slightly overzealous --- strictly speaking, there are non-well-nested PAs which accept spr-languages exclusively.
We will show in Section~\ref{sec:cfg} that checking for series-parallel rationality of a finite PA is undecidable, and must therefore accept that any decidable restriction that enforces series-parallel rationality will forbid certain valid automata.

In Section~\ref{section:expressions-to-automata}, we shall associate with every spr-expression $e$ a finitely supported \emph{and} well-nested PA that accepts $\sem{e}$.
The bi-directional correspondence between spr-expressions and pomset automata is therefore maintained.

\medskip

To ease notation, we assume for the remainder of this section that $A$ is well-nested.
We shall establish that for every state $q$ of $A$ there exists an spr-expression $e_q$ such that $L_A(q) = \sem{e_q}$.
Since $\prec_A$ is well-founded, we can proceed by induction on $\prec_A$, i.e., the induction hypothesis for $q$ is that for all $r \in Q$ with $r \prec_A q$ we can construct an $e_r \in \terms$ such that $\sem{e_r} = L_A(r)$.

The language of a recursive state is not very hard to characterise.

\begin{restatable}{lemma}{recursivestatelanguage}%
\label{lemma:recursive-state-language}
If $q \in Q_\rec$, then
\[L_A(q) = {\biggl(\bigcup\nolimits_{\gamma(q, r, q) = \top} L_A(r)\biggr)}^\dagger\]
\end{restatable}

The languages of sequential states for which our induction hypothesis holds can also be characterised.
To do this, we modify the procedure for finding a rational expression for a state in a finite automaton~\cite{mcnaughton-yamada-1960}.

\begin{definition}
Let $S \subseteq Q_{\mathsf{seq}}$, and suppose that for all $s \in S$, the induction hypothesis for $s$ holds.
For $q \in S$ and $q' \in Q$, we define $e^S_{qq'} \in \terms$, as follows.
If $q' = \bot$, we set $e_{qq'}^S = 0$.
For the remaining cases, we define $e_{qq'}^S$ inductively.
If $S = \emptyset$, then
\begin{align*}
e^S_{qq'}
    &= [q = q'] + \sum_{\delta(q, a) = q'} a + \sum_{\gamma(q, r, s) = q'} e_r \parallel e_s
\intertext{Otherwise, let $q''$ be some element of $S$, and let $S' = S \setminus \set{q''}$; then}
e^S_{qq'}
    &= e^{S'}_{qq'} + e^{S'}_{qq''} \cdot {\left(e^{S'}_{q''q''}\right)}^* \cdot e^{S'}_{q''q'}
\end{align*}
\end{definition}

Note that $e_{qq'}^\emptyset$ is well-defined, for if $\gamma(q, r, s) = q' \neq \bot$, then $r, s \prec_A q$ by the fact that $q$ is sequential, and thus $e_r$ and $e_s$ exist.
Also, the second sum is finite by the fact that $A$ is finitely supported.

\begin{restatable}{lemma}{auttoexprbottomup}%
\label{lemma:aut-to-expr-bottom-up}
Let $S \subseteq Q_{\mathsf{seq}}$, and suppose that for all $s \in S$, the induction hypothesis holds.
Let $q \in S$ and $q' \in Q$.
Then $U \in \sem{e^S_{qq'}}$ if and only if $q' \neq \bot$ and there exist $q_0, \dots, q_{\ell-1} \in S$, and $U = U_0 \cdots U_{\ell-1}$ with
\[
q = q_0 \atrace{U_0}_A q_1 \atrace{U_1}_A \dots \atrace{U_{\ell-2}}_A q_{\ell-1} \atrace{U_{\ell-1}}_A q_\ell = q'
\]
and, furthermore, for $0 \leq i < \ell$ we have that $q_i \atrace{U_i}_A q_{i+1}$ is a unit trace.
\end{restatable}

With all this in hand, we are finally ready to construct series-parallel rational expressions from pomset automata.

\begin{lemma}
If the induction hypothesis for $q$ holds, then we can construct an $e_q \in \terms$ such that $\sem{e_q} = L_A(q)$.
\end{lemma}
\begin{proof}
More generally, we show that for $q' \in Q$ with $q \preceq_A q' \preceq_A q$, we can find $e_{q'} \in \terms$ such that $L_A(q') = e_{q'}$.
We partition these states as follows
\begin{align*}
R &= \set{q' \in Q_\rec : q \preceq_A q' \preceq_A q} \\
S &= \set{q' \in Q_\seq: q \preceq_A q' \preceq_A q}
\end{align*}
Note that the induction hypothesis holds for all states in $R \cup S$: if $r \prec_A q' \preceq_A q$, then $r \prec_A q$; hence, for $r$ we can find an expression $e_r$, such that $\sem{e_r} = L_A(r)$.
Furthermore, $R$ and $S$ are finite, for $A$ is finitely supported.

We carry on to find expressions for the languages of states in $R$.
To this end, we define for $q' \in R$ that
\[e_{q'} = {\left(\sum\nolimits_{\gamma(q', r, q') = \top} e_r \right)}^\dagger\]
The above is well-defined, for if $\gamma(q', r, q') = \top$, then $r \prec_A q'$, and thus $e_r \in \terms$ exists.
Since $A$ is finitely supported, the sum is finite.
By Lemma~\ref{lemma:recursive-state-language}, we find
\[
\sem{e_{q'}}
    = {\left( \bigcup\nolimits_{\gamma(q', r, q') = \top} \sem{e_r} \right)}^\dagger
    = {\left( \bigcup\nolimits_{\gamma(q', r, q') = \top} L_A(r) \right)}^\dagger
    = L_A(q')
\]

We now consider the states in $S$.
For $q' \in S$, we define
\[
e_{q'} =
    \sum_{r \in S \cap F} e^{S}_{q'r}
    + \sum_{r \in R} e^{S}_{q'r} \cdot e_{r}
    + \sum_{r \prec_A q'} e^{S}_{q'r} \cdot e_{r}
\]
This expression is again well-defined, for all sums are finite, and $e_{r}$ exists when $r \in R$ or $r \prec_A q'$ by the above, and furthermore the induction hypothesis holds for all $r \in S$ by the observation above.

It remains to show that, for $q' \in S$, it holds that $\sem{e_{q'}} = L_A(q')$.
For the inclusion from left to right, suppose that $U \in \sem{e_{q'}}$.
There are two cases.
\begin{itemize}
    \item
    If $U \in \sem{e^{S}_{q'r}}$ for $r \in S \cap F$, then by Lemma~\ref{lemma:aut-to-expr-bottom-up} we find that $q' \atrace{U}_A r$.
    Since $r \in F$, also $U \in L_A(q')$.

    \item
    If $U \in \sem{e^{S}_{q'r} \cdot e_{r}}$ for some $r \in R$ or $r \in Q$ with $r \prec_A q'$, then $U = V \cdot W$ such that $V \in \sem{e^{S}_{q'r}}$ and $W \in \sem{e_{r}}$.
    By Lemma~\ref{lemma:aut-to-expr-bottom-up}, we find that $q' \atrace{V}_A r$; also, we find that $r \atrace{W}_A r'$ for some $r' \in F$.
    Together, this implies that $q' \atrace{V \cdot W}_A r$, and since $r \in F$ also $U = V \cdot W \in L_A(q')$.
\end{itemize}

\noindent
For the other inclusion, suppose that $U \in L_A(q')$, i.e., $q' \atrace{U} r$ for some $r \in F$.
By Lemma~\ref{lemma:trace-length}, there exist $q_0, \dots, q_n \in Q$ with $q' = q_0$ and $r = q_n$, and $U = U_0 \cdots U_{n-1}$, such that for $1 \leq i < n$ it holds that $q_i \atrace{U_i}_A q_{i+1}$.
Furthermore, each of these traces is a unit trace.
If $q_1, \dots, q_n \in S$, then $U \in \sem{e^{S}_{q'r}} \subseteq \sem{e_{q'}}$ by Lemma~\ref{lemma:aut-to-expr-bottom-up}.

Otherwise, i.e., if $q_i \not\in S$ for some $0 < i \leq n$, let $m$ be the smallest such $i$, and note that $U_m \cdots U_{n-1} \in L_A(q_m)$.
Furthermore, for $0 \leq i < m$ we have that $q_i \in S$, and thus $U_0 \cdots U_{m-1} \in \sem{e^{S}_{q'q_m}}$, by Lemma~\ref{lemma:aut-to-expr-bottom-up}.
There are two cases to consider.
\begin{itemize}
    \item
    If $q_m \in R$, then $L_A(q_m) = \sem{e_{q_m}}$ by the above.
    We conclude that
    \[
    U = U_0 \cdots U_{m-1} \cdot U_m \cdots U_{n-1}
      \in \sem{e^{S}_{q'q_m} \cdot e_{q_m}}
      \subseteq \sem{e_{q'}}
    \]

    \item
    Otherwise, if $q_m \not\in R$, then since also $q_m \not\in S$, we know that $q \not\preceq_A q_m$ or $q_m \not\preceq_A q$.
    The latter case can be excluded, for $q_m \preceq_A q' \preceq_A q$.
    We thus know that $q \not\preceq_A q_m$, and since $q \preceq_A q'$, also $q' \not\preceq_A q_m$.
    Together with $q_m \preceq_A q'$, it follows that $q_m \prec_A q'$; an argument similar to the previous case completes the proof.
    \qedhere
\end{itemize}
\end{proof}

The above establishes the main result of this section.
\begin{theorem}
Let $A$ be a well-nested and finitely supported pomset automaton.
For all states $q$ of $A$, we can find $e_q \in \terms$ such that $\sem{e_q} = L_A(q)$.
\end{theorem}

\section{Expressions to automata}%
\label{section:expressions-to-automata}

We now turn our attention to the task of constructing a pomset automaton $A$ that accepts the semantics of a given expression $e$.
Since our algorithm for obtaining expressions from a pomset automaton is sound for finitely supported and well-nested PAs only, $A$ should also satisfy these constraints.
Our approach follows Brzozowski's method for constructing a deterministic finite automaton that accepts the semantics of a rational expression~\cite{brzozowski-1964}.
More precisely, we construct a finitely supported and well-nested automaton $A_\Sigma$, such that for every spr-expression $e$ there exists a state $q_e$ such that $L_A(q_e) = \sem{e}$.
Intuitively, the transition structure of $A_\Sigma$ is set up such that the automaton can transition from the state representing $e$ to the state representing $e'$ while reading $a$ if and only if $e'$ is what ``remains'' of $e$ after consuming $a$ --- traditionally, this $e'$ is called the \emph{$a$-derivative} of $e$.

The encoding of spr-expressions into states requires some care.
Specifically, if we choose to have a state for every spr-expression, it turns out that the resulting automaton is not finitely supported.
This is not surprising; indeed, Brzozowski dealt with the same problem~\cite{brzozowski-1964}.
The solution is to encode spr-expressions into states by representing them as the equivalence classes of a congruence that is sound with respect to their semantics.

\begin{definition}%
\label{definition:congruence}
We define $\simeq$ as the smallest congruence on $\terms$ such that:
\begin{mathpar}
e + 0 \simeq e
\and%
e + e \simeq e
\and%
e + f \simeq f + e
\\
e + (f + g) \simeq (e + f) + g
\and%
(e + f) \cdot g \simeq e \cdot g + f \cdot g
\end{mathpar}
\end{definition}
Thus, when $e \simeq f$, we know that $e$ is equal to $f$, modulo associativity, commutativity and idempotence of $+$, and left-distributivity of $+$ over $\cdot$.
This congruence is decidable in polynomial time.

The set of equivalence classes of $\terms$ modulo $\simeq$ is written $\terms_\simeq$.
To lighten notation, we represent the equivalence class of $e \in \terms$ up to $\simeq$ by simply writing $e$; it will always be clear from the context whether we intend $e$ as an element of $\terms$ or $\terms_\simeq$.
We elide lemmas showing that our definitions are sound w.r.t.\ $\simeq$; arguments of this nature appear in~\ref{appendix:soundness-modulo-congruence}.

In analogy to Brzozowski's construction, where the accepting states are the rational expressions accepting the empty word, we use spr-expressions accepting the empty pomset as accepting states of our PA\@.
\begin{definition}%
\label{definition:accepting-terms}
We define the set $\sacc$ as the smallest subset of $\terms$ satisfying
\begin{mathpar}
\inferrule{~}{%
    1 \in \sacc
}
\and
\inferrule{%
    e \in \sacc \\
    f \in \terms
}{%
    e + f, f + e \in \sacc
}
\and
\inferrule{%
    e, f \in \sacc
}{%
    e \cdot f \in \sacc
}
\and
\inferrule{%
    e, f \in \sacc
}{%
    e \parallel f \in \sacc
}
\and
\inferrule{%
    e \in \terms
}{%
    e^*, e^\dagger \in \sacc
}
\end{mathpar}
\end{definition}

\begin{lemma}%
\label{lemma:nullability}
Let $e \in \terms$; then $e \in \sacc$ if and only if $1 \in \sem{e}$.
\end{lemma}

We write $\sacc_\simeq$ to denote the set of congruence classes in $\sacc$ w.r.t.\ $\simeq$.
Having identified the accepting states, we move on to the transition functions.

\begin{definition}%
\label{definition:derivatives}
Let $e, f \in \terms_\simeq$.
We use $e \star f$ to denote $f$ when $e \in \sacc_\simeq$, and $0$ otherwise; similarly, we write $e \fatsemi f$ for $0$ when $e \simeq 0$, and $e \cdot f$ otherwise.

We define the function $\ssderiv: \terms_\simeq \times \Sigma \to \terms_\simeq$ as follows:
\begin{align*}
\ssderiv(0, a) &= 0 &
    \ssderiv(e \cdot f, a) &= \ssderiv(e, a) \fatsemi f + e \star \ssderiv(f, a) \\
\ssderiv(1, a) &= 0 &
    \ssderiv(e \parallel f, a) &= 0 \\
\ssderiv(b, a) &= [a = b] &
    \ssderiv(e + f, a) &= \ssderiv(e, a) + \ssderiv(f, a) \\
\ssderiv(e^*, a) &= \ssderiv(e, a) \fatsemi e^* &
    \ssderiv(e^\dagger, a) &= 0
\intertext{%
    Furthermore, the function $\psderiv: \terms_\simeq \times \terms_\simeq \times \terms_\simeq \to \terms_\simeq$ is defined as follows:
}
\psderiv(0, g, h) &= 0 &
    \psderiv(e \cdot f,g,h ) &= \psderiv(e, g, h) \fatsemi f + e \star \psderiv(f, g, h)\\
\psderiv(1, g, h) &= 0 &
    \psderiv(e \parallel f, g,h) &= [g \simeq e \wedge h \simeq f] \\
\psderiv(b, g, h) &= 0 &
    \psderiv(e + f, g, h) &= \psderiv(e, g, h) + \psderiv(f, g, h)\\
\psderiv(e^*, g, h) &= \psderiv(e, g, h) \fatsemi e^*&
    \psderiv(e^\dagger, g, h) &= [g \simeq e \wedge h \simeq e^\dagger]
\end{align*}

We write $A_\Sigma$ for the \emph{syntactic pomset automaton}, which is $\tuple{\terms_\simeq, \ssderiv, \psderiv, {\sacc}_\simeq}$.
In this PA, the states $0$ and $1$ assume the roles of $\bot$ and $\top$ respectively.

We refer to $\ssderiv$ (respectively $\psderiv$) as the \emph{sequential} (respectively \emph{parallel}) \emph{derivative} functions.
The (strict) trace dependency relation of $A_\Sigma$ is denoted $\preceq_\Sigma$ (respectively $\prec_\Sigma$).
Similarly, the trace relation of $A_\Sigma$ is denoted by $\tracerel_\Sigma$, and we write $L_\Sigma(e)$ for the language of $e \in \terms_\simeq$ in $A_\Sigma$.
\end{definition}

We now claim that, first, $A_\Sigma$ is finitely supported and well-nested, and that, second, for $e \in \terms$ it holds that $L_\Sigma(e) = \sem{e}$.
The following two sections are devoted to showing that both of these hold, respectively.

\subsection{Structural properties}%
\label{subsection:finite-support}

Let us start by arguing that the syntactic PA is finitely supported.
To this end, we should show that for every $e \in \terms$, the set $\pi_\Sigma(e)$ is finite; since this set is the smallest closed set (w.r.t. $\preceq_\Sigma$) containing $e$, it suffices to find a finite and closed set containing $e$.
To shorten the proof, however, it is useful to introduce the following, more general notion.
\begin{definition}
Let $E \subseteq \terms$ and $e \in \terms$.
If there exist $e_0, \dots, e_{n-1} \in E$ such that $e \simeq e_0 + \cdots + e_{n-1}$, we say that $E$ \emph{covers} $e$.
Furthermore, we say that $E$ is \emph{cover-closed} if, whenever $f \in E$ and $g \preceq_\Sigma f$, it holds that $E$ covers $g$.
\end{definition}

Cover-closed sets then give us a way to show finite support, as follows.
\begin{restatable}{lemma}{coverclosedversusclosed}%
\label{lemma:cover-closed-versus-closed}
Let $e \in \terms$.
If $e$ is covered by a finite and cover-closed set, then $e$ is contained in a finite and closed set (and hence $\pi_\Sigma(e)$ is finite).
\end{restatable}

Showing finite support then comes down to finding a finite and cover-closed set for every expression, which can be done by induction on the expression.

\begin{restatable}{lemma}{syntacticpafinitelysupported}%
\label{lemma:syntactic-pa-finitely-supported}
The syntactic PA is finitely supported.
\end{restatable}

As part of the argument showing that the syntactic PA is well-nested, we need to show that $e \prec_\Sigma f$ for some spr-expressions.
To this end, it must be shown that $f \not\preceq_\Sigma e$; since it is hard to prove this directly from the inductive definition of $\preceq_\Sigma$, we introduce the following to argue $f \not\preceq_\Sigma e$ indirectly.

\begin{definition}
We define $\pdepth: \terms \to \naturals$ inductively, as follows:
\begin{align*}
\pdepth(0) &= 0 &
    \pdepth(e_0 \cdot e_1) &= \max(\pdepth(e_0), \pdepth(e_1)) \\
\pdepth(1) &= 0 &
    \pdepth(e_0 \parallel e_1) &= \max(\pdepth(e_0), \pdepth(e_1)) + 1 \\
\pdepth(a) &= 0 &
    \pdepth(e_0 + e_1) &= \max(\pdepth(e_0), \pdepth(e_1)) \\
\pdepth(e_0^*) &= \pdepth(e_0) &
    \pdepth(e_0^\dagger) &= \pdepth(e_0)
\intertext{%
    We also define $\ddepth: \terms \to \naturals$ inductively, as follows:
}
\ddepth(0) &= 0 &
    \ddepth(e_0 \cdot e_1) &= \max(\ddepth(e_0), \ddepth(e_1)) \\
\ddepth(1) &= 0 &
    \ddepth(e_0 \parallel e_1) &= \max(\ddepth(e_0), \ddepth(e_1)) \\
\ddepth(a) &= 0 &
    \ddepth(e_0 + e_1) &= \max(\ddepth(e_0), \ddepth(e_1)) \\
\ddepth(e_0^*) &= \ddepth(e_0) &
    \ddepth(e_0^\dagger) &= \ddepth(e_0) + 1
\end{align*}
\end{definition}

A straightforward series of inductive proofs on the structure of spr-expressions then gives us the following:

\begin{restatable}{lemma}{preceqvsdepth}%
\label{lemma:preceq-vs-depth}
If $e \preceq_\Sigma f$, then $d_\parallel(e) \leq d_\parallel(f)$ and $d_\dagger(e) \leq d_\dagger(f)$.
\end{restatable}
Thus, if we want to show that $e \prec_\Sigma f$, it suffices to show that $e \preceq_\Sigma f$ and $d_\parallel(e) < d_\parallel(f)$ or $d_\dagger(e) < d_\dagger(f)$.
This enables us to prove that loops involving a parallel star are trivial:
\begin{restatable}{lemma}{daggerloops}%
\label{lemma:dagger-loops}
  If $e\preceq_\Sigma f^\dagger\preceq_\Sigma e$, then $e\simeq f^\dagger$.
\end{restatable}
With this in hand, we can show the following:

\begin{restatable}{lemma}{preceqvsforks}%
\label{lemma:preceq-vs-forks}
Let $e, g, h \in \terms$ with $\psderiv(e, g, h) \not\simeq 0$.
Then $g \prec_\Sigma e$; furthermore, either $h \prec_\Sigma e$ or there exists an $f \in \terms$ such that $e \simeq f^\dagger$.
\end{restatable}

Hence, we argue that all states $A_\Sigma$ are sequential or recursive, as follows.

\begin{lemma}%
\label{lemma:syntactic-pa-well-nested}
The syntactic PA is well-nested.
\end{lemma}
\begin{proof}
Let $e \in \terms$; by Lemma~\ref{lemma:preceq-vs-forks} we already know that for all $g, h \in \terms$ such that $\psderiv(e, g, h) \not\simeq 0$ it holds that $g \prec_\Sigma e$.
If furthermore for all $g, h \in \terms$ with $\psderiv(e, g, h) \not\simeq 0$ it holds that $h \prec_\Sigma e$, then $e$ is sequential.

Otherwise, it follows by Lemma~\ref{lemma:preceq-vs-forks} that $e \simeq f^\dagger$ for some $f \in \terms$.
We now claim that, in this case, $e$ is recursive.
To see this, first note that for all $a \in \Sigma$ we have $\ssderiv(e, a) \simeq \ssderiv(f^\dagger, a) = 0$.
Furthermore, if $g, h \in \terms$ such that $\psderiv(e, g, h) \simeq \psderiv(f^\dagger, g, h) \not\simeq 0$, then $\psderiv(e, g, h) \simeq 1$ and $g \simeq f$ and $h \simeq f^\dagger$ by definition of $\psderiv$; hence $g \prec_\Sigma e$ and $h \simeq e$ by Lemma~\ref{lemma:preceq-vs-depth}.
\end{proof}

\subsection{Language equivalence}

We now set out to prove that, for $e \in \terms$, we have that $L_\Sigma(e) = \sem{e}$; to this end, we first need to discuss a number of auxiliary lemmas that help us analyse and reason about the traces in $A_\Sigma$.

For the inclusion of $L_\Sigma(e)$ in $\sem{e}$, it is useful to be able to take a trace labelled with some pomset and turn it into one or more traces labelled with (parts of) that pomset.
We refer to such an action as a \emph{deconstruction} of the starting trace.
The first deconstruction lemma that we will consider concerns traces that originate in a state that represents a sum of spr-expressions.

\begin{lemma}%
\label{lemma:trace-deconstruct-plus}
Let $e_0, e_1 \in \terms$, $f \in \sacc$ and $U \in \pomsp$, such that $e_0 + e_1 \satrace{U} f$ of length $\ell$.
There exists an $f' \in \sacc$ with $e_0 \satrace{U} f'$ or $e_1 \satrace{U} f'$ of length $\ell$.
\end{lemma}
\begin{proof}
The proof proceeds by induction $\ell$.
In the base, where $\ell = 0$, we have that $e_0 + e_1 \satrace{U} f$ is a trivial trace.
In that case, $f = e_0 + e_1$, and so $e_0 \in \sacc$ or $e_1 \in \sacc$; in the former case, choose $f' = e_0$, in the latter case, choose $f' = e_1$.
In either case, the claim is satisfied.

For the inductive step, let $e_0 + e_1 \satrace{U} f$ be of length $\ell+1$, and assume that the claim holds for $\ell$.
We find that $U = V \cdot U'$ and a $g \in \terms$ such that $e_0 + e_1 \satrace{V} g$ is a unit trace, and $g \satrace{U} f$ is of length $\ell$.
If, on the one hand, $e_0 + e_1 \satrace{V} g$ is a $\delta$-trace, then $V = a$ for some $a \in \Sigma$, and $g = \ssderiv(e_0 + e_1, a) = \ssderiv(e_0, a) + \ssderiv(e_1, a)$.
By induction, we then find $f' \in \sacc$ such that $\ssderiv(e_0, a) \satrace{U} f'$ or $\ssderiv(e_1, a) \satrace{U} f'$, of length $\ell$.
Putting this together, we have that $e_0 \satrace{U} f'$ or $e_1 \satrace{U} f'$, of length $\ell+1$.

The case where $e_0 + e_1 \satrace{V} g$ is a $\gamma$-trace is similar.
\end{proof}

The proofs of the other deconstruction lemmas follow a similar pattern; these appear in~\ref{appendix:deconstruction-lemmas}.

Another deconstruction lemma arises when the starting state is a sequential composition.
In this case, we find multiple traces: one originating in the left subterm, and another originating in the right subterm.

\begin{restatable}{lemma}{tracedeconstructsequential}%
\label{lemma:trace-deconstruct-sequential}
Let $e_0, e_1 \in \terms$, $f \in \sacc$ and $U \in \pomsp$, such that $e_0 \cdot e_1 \satrace{U} f$ is of length $\ell$.
There exist $f_0, f_1 \in \sacc$ such that $U = U_0 \cdot U_1$, as well as $e_0 \satrace{U_0} f_0$ of length $\ell_0$ and $e_1 \satrace{U_1} f_1$ of length $\ell_1$, such that $\ell_0 + \ell_1 = \ell$.
\end{restatable}

The last deconstruction lemma that we record concerns the Kleene star; here, we find a number of traces, each of which originates from the subterm under the Kleene star, and reaches an accepting state.

\begin{restatable}{lemma}{tracedeconstructstar}%
\label{lemma:trace-deconstruct-star}
Let $e \in \terms$ and $f \in \sacc$ and $U \in \pomsp$ be such that $e^* \satrace{U} f$.
There exist $f_0, \dots, f_{n-1} \in \sacc$ such that $U = U_0 \cdots U_{n-1}$ and for $0 \leq i < n$ it holds that $e \satrace{U_i} f_i$.
\end{restatable}

To show the other inclusion, i.e., that $\sem{e} \subseteq L_\Sigma(e)$, we need \emph{construction} lemmas to compose traces of pomsets into a trace of a composition of those pomsets.
To keep the lemmas concise, the following notion is convenient
\begin{definition}
We write $\lesssim$ for the smallest relation on $\terms$ such that $e \lesssim f$ when $e + f \simeq f$; note that this makes $\lesssim$ a preorder on $\terms$.
\end{definition}

The first construction lemma that we encounter allows us to use $+$ to add additional terms to the starting trace, such that the target state of the new trace contains the old target state.

\begin{lemma}%
\label{lemma:trace-construct-plus}
Let $e_0, e_1, f_0 \in \terms$ and $U \in \pomsp$ be such that $e_0 \satrace{U} f_0$.
There exists an $f \in \terms$ such that $e_0 + e_1 \satrace{U} f$ and $f_0 \lesssim f$.
\end{lemma}
\begin{proof}
The proof proceeds by induction on the length $\ell$ of $e_0 \satrace{U} f_0$.
In the base, where $\ell = 0$, we have $f_0 = e_0$ and $U = 1$.
We then choose $f = e_0 + e_1$.

For the inductive step, let $e_0 \satrace{U} f_0$ be of length $\ell+1$, and assume that the claim holds for $\ell$.
We then find $e_0' \in \terms$ and $U = V \cdot U'$ such that $e_0 \satrace{V} e_0'$ is a unit trace, and $e_0' \satrace{U'} f_0$ is of length $\ell$.
If $e_0 \satrace{V} e_0'$ is a $\delta$-trace, then $V = a$ for some $a \in \Sigma$, and $e_0' = \ssderiv(e_0, a)$.
We choose $e_1' = \ssderiv(e_1, a)$; by induction, we find $f \in \terms$ such that $f_0 \lesssim f$ and $e_0' + e_1' \satrace{U'} f$.
Since $\ssderiv(e_0 + e_1, a) = \ssderiv(e_0, a) + \ssderiv(e_1, a) = e_0' + e_1'$, we have that $e_0 + e_1 \satrace{V} e_0' + e_1'$.
Putting this together, we find that $e_0 + e_1 \satrace{U} f$.

The case where $e_0 \satrace{V} e_0'$ is a $\gamma$-trace is similar.
\end{proof}

Like deconstruction lemmas, the proofs of construction lemmas follow a similar pattern.
Further proofs of lemmas like this appear in~\ref{appendix:construction-lemmas}.

The construction lemma for sequential composition consists of two parts.
First, we need to be able to append an expression, in such a way that the appended expression is carried into the target state of the new trace.

\begin{restatable}{lemma}{traceconstructcarry}%
\label{lemma:trace-construct-carry}
Let $e_0, e_1 \in \terms$ and $f_0 \in \sacc$ and $U \in \pomsp$ be such that $e_0 \satrace{U} f_0$.
Then there exists an $f \in \terms$ such that $f_0 \cdot e_1 \lesssim f$, and $e_0 \cdot e_1 \satrace{U} f$.
\end{restatable}

Second, we need to be able to prepend an expression in $\sacc$ to get a new trace with a target state that contains the old target state.
The intuition here is that the constructed trace simply disregards the prepended expression (which is possible because it is in $\sacc$) and continues by imitating the old trace.

\begin{restatable}{lemma}{traceconstructledge}%
\label{lemma:trace-construct-ledge}
Let $e_0 \in \terms$ and $f_0, f_1 \in \sacc$ and $V \in \pomsp$ be such that $e_1 \satrace{V} f_1$.
There exists an $f \in \sacc$ such that $f_0 \cdot e_1 \satrace{V} f$.
\end{restatable}

The construction lemma for sequential composition is then a simple consequence of the preceding construction lemmas.

\begin{restatable}{lemma}{traceconstructsequential}%
\label{lemma:trace-construct-sequential}
Let $e_0, e_1 \in \terms$, $f_0, f_1 \in \sacc$ and $U, V \in \pomsp$ such that $e_0 \satrace{U} f_0$ and $e_1 \satrace{V} f_1$.
There exists an $f \in \sacc$ with $e_0 \cdot e_1 \satrace{U \cdot V} f$.
\end{restatable}

The final construction lemma shows how to construct a trace originating in a state of the form $e^*$, given a number of traces that originate in $e$.
The intuition here is that the constructed trace mimics the traces that originate in $e$, while carrying a factor $e^*$ to restart the next trace.

\begin{restatable}{lemma}{traceconstructstar}%
\label{lemma:trace-construct-star}
Let $e \in \terms$ and $f_0, \dots, f_{n-1} \in \sacc$ and $U_0, \dots, U_{n-1} \in \pomsp$ be such that for $0 \leq i < n$ it holds that $e \satrace{U_i} f_i$.
There exists an $f \in \sacc$ such that $e^* \satrace{U_0 \cdots U_{n-1}} f$.
\end{restatable}

With all of these facts about constructing and deconstructing traces in the syntactic PA, we are finally able to show correctness of our translation from expressions to automata, as witnessed by the following lemma.
\begin{restatable}{lemma}{syntacticpalanguages}%
\label{lemma:language-equivalence}
If $e \in \terms$, then $L_\Sigma(e) = \sem{e}$.
\end{restatable}

The equality follows from using the deconstruction lemmas (for the inclusion from left to right) and the construction lemmas (to show the inclusion from right to left); as before, a full proof can be found in~\ref{appendix:expressions-to-automata}.

This establishes the main result of this section.
\begin{theorem}
For $e \in \terms$, we can find a well-nested and finitely supported PA $A$ that accepts $\sem{e}$, i.e., $A$ has a state $q_e$ such that $L_A(q_e) = \sem{e}$.
\end{theorem}

\section{Context-free pomset languages}%
\label{sec:cfg}
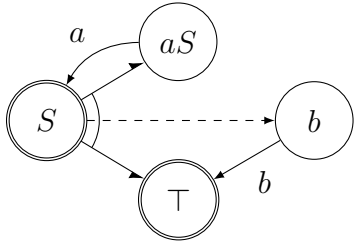
\begin{figure}
    \centering
    \begin{tikzpicture}[scale=2]
      \node[state,accepting] (q1) {$S$};
      \node[state,accepting,below right=3mm and 10mm of q1] (q3) {$\top$};
      \node[state,above right=3mm and 10mm of q1] (q2) {$aS$};
      \node[state,right=25mm of q1] (q4) {$b$};

      \draw[ptransleft] (q1) edge (q2);
      \draw[ptransright] (q1) edge (q3);
      \draw pic[draw,angle radius=.7cm] {angle=q3--q1--q2};
      \draw[dashed,-latex] (q1) edge (q4);
      \draw[-latex] (q4) edge node[below right] {$b$} (q3);
      \draw[-latex] (q2) edge[bend right] node[above left] {$a$} (q1);
    \end{tikzpicture}
    \caption{The PA $A_2$, recognising $a^{n}b^{n}$}\label{figure:problem-cfg}
\end{figure}
In this section, we characterise the class of languages accepted by finite PA, with no restrictions.
These turn out to be languages of pomsets generated by finite context-free grammars~\cite{chomsky-1956} using series-parallel terms.
A pomset automaton whose language is not rational is displayed in Figure~\ref{figure:problem-cfg}.

A \emph{context-free pomset grammar} $G$ (\emph{CFG}) is given by a triple $\tuple{\Gamma,S,R}$, where $\Gamma$ is a finite set of non-terminals, $S\in \Gamma$ is a distinguished start symbol, and $R$ is a finite set of production rules, i.e., pairs of a non-terminal and a term built out of sequential products, parallel products, and symbols chosen from $\Gamma\cup\Sigma\cup\set\epsilon$.
Using the production rules as usual starting from the symbol $S$, we define the pomset language $\sem G$ generated by a CFG~$G$.
A pomset language is called \emph{context-free} (\emph{CF}) if it is generated by some CFG\@.

\begin{theorem}%
\label{sec:auto-to-cfg}
A pomset language is accepted by a PA if and only if it is CF\@.
\end{theorem}
\begin{proof}
  The automaton to grammar direction is straightforward.
  Given a finite PA~$A=\tuple{Q,\delta,\gamma,F}$ and $q_0\in Q$, we will build a CFG $G_{A,q_0}$ with non-terminals~$Q$ and start symbol $q_0$.
  For every state $q$ and letter $a\in\Sigma$ we produce a rule $q\to a\cdot\delta(q,a)$; we add for every triple of states $(q,r,s)$ a production $q\to (r\parallel s)\cdot\gamma(q,r,s)$; finally for every accepting state $q \in F$ we add a rule~$q\to\epsilon$.
  The fact that $\sem{G_{A,q_0}}=L_A(q_0)$ is straightforward from this definition: clearly, there is a correspondence between the accepting runs of~$A$ starting from $q_0$ and the derivations in $G_{A,q_0}$.

  We now construct an automaton from a CFG $G=\tuple{\Gamma,S,R}$ to recognise $\sem{G}$.
  Let $\mathcal T$ be the set of subterms of the right-hand sides of rules in $R$, $s,t$ will range over $\mathcal T$ in the following;
  this set is clearly finite.
  We define a PA $A_G=\tuple{\Gamma\cup \mathcal T\cup\set{\top,\bot},\delta,\gamma,\set{\top,\epsilon}}$, such that $L_A(S)=\sem G$ where:
  \begin{mathpar}
    \delta(a,a):=\top;\and
    \gamma(s\parallel t,s,t):=\top;
    \\
    \gamma(s\cdot t,s,\top):=t;\and
    \gamma(X,s,\top):=
      \begin{cases}
        \top    & \mbox{if}\ X \to s \in R \\
        \bot    & \mbox{otherwise}
      \end{cases}
  \end{mathpar}
  We complete $\delta$ and $\gamma$ into functions by assigning every undefined value to $\bot$.
  There is a straightforward correspondence between runs in $A_G$ and derivations from $G$, therefore they are language equivalent.
\end{proof}
As usual when in the presence of context-free languages~\cite{bar-hillel-perles-shamir-1961,greibach-1968}, we obtain a host of undecidability results for PA with no restriction other than finiteness; we call out two important ones below.
\begin{corollary}
Let $A$ be a PA\@.
The following are undecidable:
\begin{enumerate}[(i)]
    \setlength{\itemsep}{0em}
    \item
    Given states $q$ and $q'$ of $A$, does $L_A(q) = L_A(q')$ hold?
    \item
    Given a state $q$ of $A$, is $L_A(q)$ an spr-language?
\end{enumerate}
\end{corollary}
The second undecidability result justifies our ``well-nestedness'' condition for the automata-to-expression direction of our Kleene theorem, since one needs strict restrictions on PA to guarantee the rationality of its language.

\section{Related work}%
\label{section:related-work}

If a PA is fork-acyclic in the sense of~\cite{kappe-brunet-luttik-silva-zanasi-2017}, then it is well-nested; thus, finitely supported and well-nested PAs are a superset of the PAs considered in~\cite{kappe-brunet-luttik-silva-zanasi-2017}.
Relaxing fork-acyclicity to well-nestedness is necessary for finitely supported pomset automata to capture spr-expressions that contain $\dagger$.

Lodaya and Weil proposed another automaton formalism for pomsets, called \emph{branching automata}~\cite{lodaya-weil-2000}.
These define states where parallelism can start (\emph{fork}) or end (\emph{join}) in two relations; pomset automata condense this information in a single function.
In op.\ cit., we also find a translation of spr-expressions to branching automata, based on Thompson's construction~\cite{thompson-1968}, which relies on the fact that transitions of branching automata are encoded as \emph{relations}.
Our Brzozowski-style~\cite{brzozowski-1964} translation, in contrast, constructs transition \emph{functions} from the expressions.
Lastly, translation of branching automata to series-parallel expressions in~\cite{lodaya-weil-2000} is sound only for a \emph{semantically} restricted class of automata, whereas our restriction is \emph{structural}.

Jipsen and Moshier~\cite{jipsen-moshier-2016} provided an alternative formulation of the automata proposed by Lodaya and Weil, also called \emph{branching automata}.
Their method to encode parallelism is conceptually dual to pomset automata: branching automata distinguish based on the target states of traces to determine the join state, whereas pomset automata distinguish based on the origin states of traces.
The translations of series-parallel expressions to branching automata and vice versa suffer from the same shortcomings as those by Lodaya and Weil, i.e., transition relations rather than functions and a semantic restriction on automata for the translation of automata to expressions.

\emph{Series-rational expressions} are spr-expressions that do not use the connective $\dagger$.
Lodaya and Weil described~\cite{lodaya-weil-2000} a fragment of their branching automata whose languages match series-rational languages, and whose behaviour corresponds to $1$-safe Petri nets.
This fragment can be matched with a fragment of pomset automata (discussed in~\cite{kappe-brunet-luttik-silva-zanasi-2017}).

We opted to treat semantics of spr-expressions in terms of automata instead of Petri nets to find more opportunities to extend to a coalgebraic treatment.
The present paper does not reach this goal, but we believe that our formulation in terms of states and transition functions offers some hope of getting there.
On the other hand, the Petri net perspective allows for equivalence-checking algorithms.
Brunet, Pous, and Struth~\cite{brunet-pous-struth-2017} used $1$-safe Petri nets to provide an algorithm to compare sr-expressions.
They also provided an algorithm to compare such expressions with respect to another semantics, the \emph{downward-closed} semantics.
Both these algorithms run in exponential space, and the second problem was shown in op.\ cit.\ to be complete for this complexity class.

\emph{Petri automata}~\cite{brunet-pous-2017} are yet another class of automata for series-rational languages.
These Petri net-based automata recognise the languages of series-parallel graphs that can be denoted by series-rational expressions.
The connection between series-parallel pomsets and series-parallel graphs is achieved through duality: a series-parallel pomset is the pomset of edges of some series-parallel graph, and vice-versa.
For CKA the pomset point of view is convenient because the exchange law may be expressed much more naturally on pomsets than on graphs.
On the other hand, Petri automata were introduced to investigate another class of algebras, namely \emph{Kleene allegories}~\cite{brunet-pous-2015}, where the free semantics is expressed in terms of graph homomorphism, making the graph view the natural choice.

Also related are \emph{parenthesising automata} as proposed by \'{E}sik and N\'{e}meth in~\cite{esik-nemeth-2004}, which recognise \emph{series-parallel $n$-posets}, a generalisation of words where events are partially ordered by $n$ partial orders.
Like sp-pomsets, series-parallel $2$-posets can be composed using two associative operators, but unlike sp-pomsets, the ``parallel'' composition operator is not commutative.
This does not rule out a connection to programs with parallelism, but it does require specialisation of the model.
On the other hand, parenthesising automata have an advantage over pomset automata in that they are pleasingly symmetric in how composition operators are treated, which simplifies a lot of proofs.
The correspondence between automata and expressions described in op.\ cit.\ also requires restricting the class of automata.
Unlike our work, however, this restriction tightly characterises the automata for which the translation is possible, and is furthermore decidable.
Since parenthesising automata cannot, in general, recognise context-free languages, this does not contradict our earlier remarks about decidability of such a property.

Prisacariu introduced \emph{Synchronous Kleene Algebra} (SKA)~\cite{prisacariu-2010}, extending Kleene Algebra with a \emph{synchronous composition} operator.
SKA differs from our model in that it assumes that all basic actions are performed in unit time, and that actors responsible for individual actions never idle.
In contrast, our (BKA-like) model makes no synchrony assumptions: expressions can be composed in parallel, and the relative timing of basic actions within those expressions is irrelevant for the semantics.
Prisacariu axiomatised SKA and extended it to \emph{Synchronous Kleene Algebra with Tests} (SKAT); others proposed Brzozowski-style derivatives of SKA- and SKAT-expressions~\cite{broda-cavadas-ferreira-moreira-2015}.

\section{Further work}%
\label{section:further-work}

Language equivalence of rational expressions can be axiomatised using Kleene's original theorem~\cite{kozen-1994}.%
\footnote{A similar result exists for spr-expressions~\cite{laurence-struth-2014}, but this does not rely on Kleene's theorem for canonicalisation.}
More precisely, the proof in op.\ cit.\ relies on encoding a minimised finite automaton for a rational expression back into a rational expression (using both directions of Kleene's theorem) to obtain an equivalent canonical representation.
We hope to apply the work put forward in the present paper to axiomatise spr-expressions in the same fashion.
In particular, the correspondence of expressions to states and the structural nature of well-nestedness may prove useful in validating such a canonicalisation.
For this technique to work, one would need to devise a canonical form for PAs, analogous to the minimal finite automaton.

A different result axiomatises equivalence of sr-expressions (i.e., spr-expressions without the parallel star) with respect to the downward-closed pomset semantics~\cite{laurence-struth-2017-arxiv,kappe-brunet-silva-zanasi-2018}.
The algorithm in~\cite{kappe-brunet-silva-zanasi-2018} for constructing the downward closure of an sr-expression is particularly relevant as it can be used to extend the direction from expressions to automata of our Kleene theorem.
More precisely, it establishes pomset automata as an operational model for \emph{weak CKA}, that is, BKA without the parallel star $\dagger$ but with the exchange law.
Extending the result even further to spr-expressions is not possible with the methodology used in~\cite{kappe-brunet-silva-zanasi-2018} or~\cite{laurence-struth-2017-arxiv}, see the conclusions of~\cite{kappe-brunet-silva-zanasi-2018}.

Brzozowski derivatives for classic rational expressions induce a coalgebra on rational expressions that corresponds to a finite automaton.
We aim to study spr-expressions coalgebraically.
The first step would be to find the coalgebraic analogue of pomset automata such that language acceptance is characterised by the homomorphism into the final coalgebra.
Ideally, such a view of pomset automata would give rise to a decision procedure for equivalence of spr-expressions based on coalgebraic bisimulation-up-to~\cite{rot-bonsangue-rutten-2013}.

Rational expressions can be extended with tests to reason about imperative programs equationally~\cite{kozen-1997}.
In the same vein, one can extend sr-expressions with tests~\cite{jipsen-2014,jipsen-moshier-2016} to reason about parallel imperative programs equationally.
We are particularly interested in employing such an extension to extend the network specification language NetKAT~\cite{anderson-foster-guha-etal-2014} with primitives for concurrency so as to model and reason about concurrency within networks.

\bibliographystyle{elsarticle-num}
\bibliography{bibliography}

\newpage

\appendix

\section{Proofs about pomset automata}

\tracetopbottom*
\begin{proof}
The proof proceeds by induction on the construction of $\tracerel_A$.
In the base, $U = a$ for some $a \in \Sigma$, and $q' = \delta(q, a)$.
It follows that $q' = \bot$.

For the inductive step, there are two cases to consider.
\begin{itemize}
    \item
    If $q \atrace{U}_A q'$ because $U = V \cdot W$ and there exists a $q'' \in Q$ such that $q \atrace{V}_A q''$ and $q'' \atrace{W}_A q'$, then we can assume without loss of generality that at least one of these traces is non-trivial --- if this were not the case, then $q = q'' = q$ and $V = W = 1$, meaning that $q = q'$ and $U = 1$, and so $q \atrace{U}_A q'$ would be trivial as well.

    If, on the one hand, $q \atrace{V}_A q''$ is non-trivial, then $q'' = \bot$ by induction.
    In that case, $q' = \bot$, regardless of whether $q'' \atrace{W}_A q'$ is trivial.
    If, on the other hand, $q'' \atrace{W}_A q'$ is non-trivial, then $q' = \bot$, also by induction.

    \item
    If $q \atrace{U}_A q'$ is a $\gamma$-trace, i.e., $U = V \parallel W$ and there exist $r, s \in Q$ as well as $r', s' \in Q$ such that $r \atrace{V}_A r'$ and $s \atrace{W}_A s'$ and $\gamma(q, r, s) = q'$, then $q' = \bot$ immediately.
    \qedhere
\end{itemize}
\end{proof}

\restrictfinitelysupportedautomaton*
\begin{proof}
We define $A_q$ as the automaton $\tuple{\pi_A(q), \delta, \gamma, F \cap \pi_A(q)}$.
Here, $\delta$ and $\gamma$ are well-defined as functions of type $\pi_A(q) \times \Sigma \to \pi_A(q)$ and ${\pi_A(q)}^3 \to \pi_A(q)$ respectively, by definition of $\pi_A(q)$.
Furthermore, since $A$ is finitely supported, we know that $A_q$ has finitely many states.
It remains to show that $L_{A_q}(q) = L_A(q)$.

For the inclusion from left to right, we prove the more general claim that if $r \atrace{U}_{A_q} r'$ with $r' \neq \bot$, then $r \atrace{U}_A r'$.
The proof proceeds by induction on the construction of $\tracerel_{A_q}$.
In the base, there are two cases to consider.
\begin{itemize}
    \item
    If $U = 1$ and $r = r'$, then $r \atrace{U}_A r'$ immediately.

    \item
    If $U = a$ for some $a \in \Sigma$ and $r' = \delta(r, a)$, then it follows that $r \atrace{U}_A r'$.
\end{itemize}

\noindent
For the inductive step, there are two cases to consider.
\begin{itemize}
    \item
    Suppose that $r \atrace{U}_{A_q} r'$ because $U = V \cdot W$ and there exists an $r'' \in \pi_A(q)$ with $r \atrace{V}_{A_q} r''$ and $r'' \atrace{W}_{A_q} r'$.
    Since $r' \neq \bot$, also $r'' \neq \bot$ by Lemma~\ref{lemma:trace-top-bottom}.
    By induction, we find that $r \atrace{V}_A r'' \atrace{W}_A r'$, and thus $r \atrace{U}_A r'$.

    \item
    Suppose that $r \atrace{U}_{A_q} r'$ because $U = V \parallel W$ and there exist $s, t \in \pi_A(q)$ and $s', t' \in \pi_A(q) \cap F$ such that $s \atrace{V}_{A_q} s'$ and $t \atrace{W}_{A_q} t'$ and $\gamma(r, s, t) = r'$.
    First, note that $s', t' \neq \bot$, since $s', t' \in F$.
    By induction, we find that $s \atrace{V}_A s' \in F$ and $t \atrace{W}_A t' \in F$.
    We can then conclude that $q \atrace{U}_A q'$.
\end{itemize}

\noindent
For the other inclusion, we prove the more general claim that if $r \in \pi_A(q)$ and $r \atrace{U}_A r'$ with $r' \neq \bot$, then $r' \in \pi_A(q)$ and $r \atrace{U}_{A_q} r'$.
We proceed by induction on the construction of $\tracerel_A$.
In the base, there are two cases to consider.
\begin{itemize}
    \item
    If $U = 1$ and $r = r'$, then $r' \in \pi_A(q)$ and $r \atrace{U}_{A_q} r'$ immediately.

    \item
    If $U = a$ for some $a \in \Sigma$, and $r' = \delta(r, a)$, then note that $r' \in \pi_A(q)$ by definition of $\pi_A(q)$.
    Furthermore, we find that $r \atrace{U}_{A_q} r'$.
\end{itemize}

\noindent
For the inductive step, there are two cases to consider.
\begin{itemize}
    \item
    Suppose $r \atrace{U}_A r'$ because $U = V \cdot W$, and there exists an $r'' \in Q$ such that $r \atrace{V}_A r''$ and $r'' \atrace{W}_A r'$.
    By induction, we then find that $r'' \in \pi_A(q)$ and $r \atrace{V}_{A_q} r''$.
    Again by induction, we also find that $r' \in \pi_A(q)$ and $r'' \atrace{W}_{A_q} r'$.
    We then conclude that $r \atrace{U}_{A_q} r'$.

    \item
    Suppose $r \atrace{U}_A r'$ because $U = V \parallel W$, and there exist $s, t \in Q$ and $s', t' \in F$ such that $s \atrace{V}_A s'$ and $t \atrace{W}_A t'$ and $\gamma(r, s, t) = r'$.
    By the premise that $r' \neq \bot$ we have that $s, t \preceq_A q$, and thus $s, t \in \pi_A(q)$.
    By induction, we then find that $s', t' \in \pi_A(q)$, and $s \atrace{V}_{A_q} s'$ as well as $t \atrace{W}_{A_q} t'$.
    We can then conclude that $r \atrace{U}_{A_q} r'$.
    \qedhere
\end{itemize}
\end{proof}

\finitelysupportedimplieswellfounded*
\begin{proof}
Suppose, towards a contradiction, that $\set{q_n}_{n \in \naturals} \subseteq Q$ is such that for $n \in \naturals$ it holds that $q_{n+1} \prec_A q_n$.
Since $\set{q_n}_{n \in \naturals} \subseteq \pi_A(q_0)$ and the latter is finite, it follows that $q_n = q_m$ for some $n > m$.
But then we find that, $q_n \prec_A q_m$, which contradicts that $\prec_A$ is a strict order, and therefore irreflexive.
\end{proof}

\tracelength*
\begin{proof}
The proof proceeds by induction on the construction of $q \atrace{U}_A q'$.
In the base, there are two cases to consider.
\begin{itemize}
    \item
    If $q \atrace{U}_A q'$ is a trivial trace, then the claim holds immediately; simply choose $\ell = 0$ and $q_0 = q = q'$.

    \item
    If $q \atrace{U}_A q'$ is a $\delta$-trace, i.e., $U = a$ for some $a \in \Sigma$ and $q' = \delta(q, a)$, then choose $\ell = 1$ and $U_0 = U = a$ to satisfy the claim.
\end{itemize}

\noindent
In the inductive step, there are again two cases to consider.
\begin{itemize}
    \item
    Suppose that $q \atrace{U}_A q'$ because $U = V \cdot W$ and there exists a $q'' \in Q$ such that $q \atrace{V}_A q''$ and $q'' \atrace{W}_A q'$.
    By induction, we find $q_0, \dots, q_n \in Q$ with $q_0 = q$ and $q_n = q''$, and $V = V_0 \cdots V_{n-1}$ such that for $0 \leq i < \ell'$ it holds that $q_i \atrace{V_i}_A q_{i+1}$.
    Also by induction, we find $q_0', \dots, q_m'$ with $q_0' = q''$ and $q_m' = q'$, and $W = W_0 \cdots W_{m-1}$ such that for $0 \leq i < \ell''$ it holds that $q_i' \atrace{W_i}_A q_{i+1}'$.
    We then choose for $0 \leq i < \ell'$ that $U_i = V_i$, and for  $\ell' \leq i < \ell' + \ell''$ that $q_{i+\ell'} = q_i'$ and $U_{i+\ell'} = W_i$ to satisfy the claim.

    \item
    Suppose that $q \atrace{U}_A q'$ is a $\gamma$-trace, i.e., $U = V \parallel W$ and there exist $r, s \in Q$ and $r', s' \in F$ such that $r \atrace{V}_A r'$ and $s \atrace{W}_A s'$ and $\gamma(q, r, s) = q'$.
    In that case, we can choose $n = 1$ and $U_0 = U = V \parallel W$ to satisfy the claim.
    \qedhere
\end{itemize}
\end{proof}

\section{Automata to expressions}

\recursivestatelanguage*
\begin{proof}
For brevity, we write $L'$ for the right-hand side of the claimed equality.

For the inclusion from left to right, we prove the more general claim that if $q \atrace{U}_A q'$ for some $q' \in F$, then $U \in L'$ and furthermore $q' \in \set{\top, q}$.
The proof proceeds by induction on the construction of $\tracerel_A$.
In the base, there are two cases to consider.
\begin{itemize}
    \item
    If $q \atrace{U}_A q'$ because $U = 1$ and $q = q'$, then the claim follows immediately.

    \item
    If $q \atrace{U}_A q'$ because $U = a$ for some $a \in \Sigma$ and $q' = \delta(q, a)$, then $q' = \bot$ by the premise that $q$ is recursive.
    Therefore, we can disregard this case, because it contradicts the premise that $q' \in F$.
\end{itemize}

\noindent
For the inductive step, there are again two cases to consider.
\begin{itemize}
    \item
    If $q \atrace{U}_A q'$ because $U = V \cdot W$ and there exists a $q'' \in Q$ with $q \atrace{V}_A q''$ and $q'' \atrace{W}_A q'$, then there are two subcases to consider.
    \begin{itemize}
        \item
        If $q \atrace{V}_A q''$ is trivial, then the claim follows by applying the induction hypothesis to $q = q'' \atrace{W}_A q'$, noting that $U = V \cdot W = W$.

        \item
        If $q \atrace{V}_A q''$ is non-trivial, then $q'' \in \{ \bot, \top \}$ by the premise that $q$ is recursive.
        Since $q' \in F$, it then follows that $q'' = \bot$ and $q'' \atrace{W}_A q'$ is trivial, by Lemma~\ref{lemma:trace-top-bottom}.
        The claim then follows by applying the induction hypothesis to $q \atrace{V}_A q''$, since $U = V \cdot W = V$.
    \end{itemize}

    \item
    If $q \atrace{U}_A q'$ because $U = V \parallel W$ and there exist $r, s \in Q$ and $r', s' \in F$ such that $r \atrace{V}_A r'$ and $s \atrace{W}_A s'$ and $\gamma(q, r, s) = q'$, then, since $q' \neq \bot$, it follows that $q' = \top$ and $s = q$ by the premise that $q$ is recursive.
    By induction, we then know that $W \in L'$.
    Furthermore, $V \in L_A(r)$ by definition.
    Consequently, $U = V \parallel W \in L_A(r) \parallel L' \subseteq L'$.
\end{itemize}

\noindent
For the inclusion from right to left, let $U \in L'$.
Then $U = U_0 \parallel \dots \parallel U_{n-1}$ and for $0 \leq i < n$ there exists an $r_i \in Q$ such that $\gamma(q, r_i, q) = \top$ and $U_i \in L_A(r_i)$.
We need to prove that $U \in L_A(q)$, which we do by induction on $n$.
In the base, where $n = 0$, we have that $U = 1$, and thus $U \in L_A(q)$ immediately, for $q \in F$.
For the inductive step, assume that $n > 0$ and that the claim holds for $n-1$; then $U' = U_1 \parallel \dots \parallel U_{n-1} \in L_A(q)$ by induction.
We thus find a $q' \in F$ such that $q \atrace{U'}_A q'$.
Furthermore, since $U_0 \in L_A(r_0)$, we find an $r_0' \in F$ such that $r_0 \atrace{U_0} r_0'$.
We then know that $q \atrace{U_0 \parallel U'} \gamma(q, r_0, q) = \top$, and hence $U = U_0' \parallel U' \in L_A(q)$.
\end{proof}

\auttoexprbottomup*
\begin{proof}
For the direction from left to right, first note that if $q' = \bot$, then $e_{qq'}^S = 0$, meaning $\sem{e_{qq'}^S} = \emptyset$; consequently, $q' \neq \bot$.
For the remainder, we proceed by induction on $S$.
In the base, where $S = \emptyset$, we have three cases.
\begin{itemize}
    \item
    If $U = 1$ and $q = q'$, then we can choose $\ell = 0$ to satisfy the claim.

    \item
    If $U = a$ for $a \in \Sigma$ with $\delta(q, a) = q'$, then we choose $\ell = 1$ and $U_0 = a$.

    \item
    If $U = V \parallel W$ and $V \in \sem{e_r}$ and $W \in \sem{e_s}$ with $\gamma(q, r, s) = q'$, then $r, s \prec_A q$.
    By induction, $V \in L_A(r)$ and $W \in L_A(s)$; therefore, there exist $r', s' \in F$ such that $r \atrace{V}_A r'$ and $s \atrace{W}_A s'$.
    We then again choose $\ell = 1$ and $U_0 = U = V \parallel W$ to find that $q \atrace{U}_A q'$.
\end{itemize}

\noindent
For the inductive step, let $S = S' \cup \set{q''}$, and assume the claim holds for $S'$.
There are two cases to consider.
\begin{itemize}
    \item
    If $U \in \sem{e^{S'}_{qq'}}$, then the claim follows by induction.

    \item
    If $U \in \sem{e^{S'}_{qq''} \cdot {\left(e^{S'}_{q''q''}\right)}^* \cdot e^{S'}_{q''q'}}$, then $U = V \cdot W_0 \cdots W_{m-1} \cdot X$ with
    \begin{mathpar}
    V \in \semnostretch{e^{S'}_{qq''}} \and
    W_0 \in \semnostretch{e^{S'}_{q''q''}} \and
    \cdots \and
    W_{m-1} \in \semnostretch{e^{S'}_{q''q''}} \and
    X \in \semnostretch{e^{S'}_{q''q'}}
    \end{mathpar}
    It should be obvious how to construct the desired trace.
\end{itemize}

\noindent
For the other direction, first note that since $q' \neq \bot$, we know by Lemma~\ref{lemma:trace-top-bottom} that $q_0, \dots, q_{\ell-1} \neq \bot$.
The proof proceeds by induction on $\ell$.
In the base, where $\ell \leq 1$, there are three cases to consider.
\begin{itemize}
    \item
    If $q \atrace{U} q'$ is trivial, then $U = 1$ and $q = q'$; thus $U \in \sem{e^\emptyset_{qq'}} \subseteq \sem{e^S_{qq'}}$.

    \item
    If $q \atrace{U} q'$ is a $\delta$-trace, then $U = a$ for some $a \in \Sigma$ and $\delta(q, a) = q'$.
    We find that $U = a \in \sem{e^\emptyset_{qq'}} \subseteq \sem{e^S_{qq'}}$.

    \item
    If $q \atrace{U} q'$ is a $\gamma$-trace, then $U = V \parallel W$ with $r, s \in Q$ and $r', s' \in F$ such that $r \atrace{V}_A r'$ and $s \atrace{W}_A s'$ and $\gamma(q, r, s) = q'$, then $r, s \prec_A q$.
    By induction we have $V \in \sem{e_r}$ and $W \in \sem{e_s}$.
    Therefore,
    \[U = V \parallel W \in \semnostretch{e_r \parallel e_s} \subseteq \semnostretch{e^\emptyset_{qq'}} \subseteq \semnostretch{e^S_{qq'}}\]
\end{itemize}

\noindent
For the inductive step, assume that $\ell > 1$; in that case, it must be that $S \neq \emptyset$.
We write $S = S' \cup \set{q''}$ and $I = \set{0 \leq i < \ell : q_i = q''}$.
If $I = \emptyset$, then $U \in \sem{e^{S'}_{qq'}}$ by induction.
Since $\sem{e^{S'}_{qq'}} \subseteq \sem{e^S_{qq'}}$, the claim follows.
Otherwise, if $I \neq \emptyset$, then write $I = \set{i_0, \dots, i_{k-1}}$ with $i_0 < i_1 < \cdots < i_{k-1}$.
Then, by induction, we know that for $1 \leq j < k$ it holds that
\[
U_{i_j} \cdot U_{i_j+1} \cdots U_{i_{j+1}-1} \in \semnostretch{e^{S'}_{q''q''}}
\]
Moreover, also by induction, we know that
\begin{mathpar}
U_0 \cdots U_{i_1-1} \in \semnostretch{e^{S'}_{qq''}} \and
U_{i_{k-1}} \cdot U_{i_{k-1}+1} \cdots U_{\ell} \in \semnostretch{e^{S'}_{q''q'}}
\end{mathpar}
Putting this together, we have
\[
U \in \sem{e^{S'}_{qq''} \cdot {\left( e^{S'}_{q''q''} \right)}^* \cdot e^{S'}_{q''q'}} \subseteq \semnostretch{e^S_{qq'}}
\qedhere
\]
\end{proof}

\section{Expressions to automata}%
\label{appendix:expressions-to-automata}

\subsection{Finite support}

\coverclosedversusclosed*
\begin{proof}
We choose $F = \{ e_0 + \cdots + e_{n-1} : e_0, \dots, e_{n-1} \in E \}$; now $F$ is finite.
Since $E$ covers $e$, it also follows that $F$ contains $e$.
To see that $F$ is closed, let $e' \simeq e_0' + \cdots + e_{n-1}' \in F$ for $e_0', \dots, e_{n-1}' \in E$, and suppose $f \preceq_\Sigma e'$.
To see that $f \in F$, it suffices to validate the claim for the pairs generating $\preceq_\Sigma$:
\begin{itemize}
    \item
    If $f = \ssderiv(e', a)$ for $a \in \Sigma$, then $f \simeq f_0 + \cdots + f_{n-1}$ where for $0 \leq i < n$ we have $f_i = \ssderiv(e_i', a)$.
    Each of these $f_i$ is covered by $E$; hence, the sum of terms covering these is in $F$, and congruent to $f$.
    The case where $f = \psderiv(e', g, h)$ for $g, h \in \terms$ can be argued similarly.

    \item
    If $f \preceq_\Sigma e'$ because $\psderiv(e', f, g) \not\simeq 0$ or $\psderiv(e', g, f) \not\simeq 0$ for some $g \in \terms$, then (without loss of generality) assume the former.
    We then know that $\psderiv(e_i', f, g) \not\simeq 0$ for some $0 \leq i < n$, and hence $f \preceq_\Sigma e_i'$, meaning that there exist $f_0, \dots, f_{m-1} \in E$ such that $f \simeq f_0 + \cdots + f_{m-1}$, by cover-closure of $E$.
    It then follows that $f \in F$.
    \qedhere
\end{itemize}
\end{proof}

\syntacticpafinitelysupported*
\begin{proof}
By Lemma~\ref{lemma:cover-closed-versus-closed}, it suffices to find for every $e \in \terms_\simeq$ a finite and cover-closed set $E(e)$ covering $e$.
We proceed inductively.
In the base, there are three cases to consider.
\begin{itemize}
    \item
    If $e = 0$, then $E(0) = \emptyset$ suffices, since $0$ is covered by the empty sum.

    \item
    If $e = 1$, then $E(1) = \set{1}$ suffices.

    \item
    If $e = a$ for some $a \in \Sigma$, then $E(a) = \set{1, a}$ suffices.
\end{itemize}

\noindent
For the inductive step, there are five cases to consider.
\begin{itemize}
    \item
    If $e = f + g$, we choose $E(e) = E(f) + E(g)$.
    This set is cover-closed, because $E(f)$ and $E(g)$ both are.
    Furthermore, this set covers $e$, because we can get terms to cover $f$ and $g$ from $E(f)$ and $E(g)$ respectively.

    \item
    If $e = f \cdot g$, we choose $E(e) = E(f) \cup E(g) \cup \set{f' \cdot g : f' \in E(f)}$.
    To see that this set covers $f \cdot g$, let $f_0, \dots, f_{n-1} \in E(f)$ be such that $f \simeq f_0 + \cdots f_{n-1}$.
    We can then choose $f_0 \cdot g, \dots, f_{n-1} \cdot g \in E(e)$ to find that $e \simeq f_0 \cdot g + \cdots f_{n-1} \cdot g$.

    To see that $E(e)$ is cover-closed, we need only consider $f' \cdot g$ for $f' \in E(f)$.
    \begin{enumerate}[(i)]
        \item
        If $a \in \Sigma$, let $f_0', \dots, f_{n-1}' \in E(f)$ with $\ssderiv(f', a) \simeq f_0' + \cdots + f_{n-1}'$.
        Also, let $g_0, \dots, g_{m-1} \in E(g)$ with $\ssderiv(g, a) \simeq g_0 + \cdots + g_{m-1}$.
        We can then derive as follows:
        \begin{align*}
        \ssderiv(e, a)
            &= \ssderiv(f', a) \fatsemi g + f' \star \ssderiv(g, a) \\
            &\simeq (f_0' + \cdots + f_{n-1}) \fatsemi g + f' \star (g_0 + \cdots + g_{m-1}) \\
            &\simeq f_0' \cdot g + \cdots + f_{n-1}' \cdot g + f' \star g_0 + \cdots f' \star g_{m-1}
        \end{align*}
        All of the non-zero terms in the last form can be found in $E(e)$, and thus $E(e)$ covers $\ssderiv(e, a)$.

        \item
        If $h_1, h_2 \in \terms_\simeq$, then $\psderiv(f' \cdot g, h_1, h_2)$ is covered by $E(e)$ by an argument similar to the above.

        \item
        If $h_1, h_2 \in \terms_\simeq$ such that $\psderiv(f' \cdot g, h_1, h_2) \not\simeq 0$, then
        \[
            \psderiv(f', h_1, h_2) \cdot g + f' \star \psderiv(g, h_1, h_2) \not\simeq 0
        \]
        Hence, we know that either $\psderiv(f', h_1, h_2) \not\simeq 0$ or $\psderiv(g, h_1, h_2) \not\simeq 0$.
        In the former case, $h_1$ and $h_2$ are covered by $E(f)$, while in the latter case $h_1$ and $h_2$ are covered by $E(g)$.
    \end{enumerate}

    \item
    If $e = f \parallel g$, we choose $E(g) = E(e) \cup E(f) \cup \set{1, f \parallel g}$.
    Immediately, $E(e)$ covers $e$.
    For cover-closure of $E(e)$, we need only consider $f \parallel g$.
    \begin{enumerate}[(i)]
        \item
        If $a \in \Sigma$, then $\ssderiv(e, a) = 0$, and so $E(e)$ covers $\ssderiv(e, a)$.

        \item
        If $h_1, h_2 \in \terms_\simeq$, then $\psderiv(e, h_1, h_2) \in \set{0, 1}$ by definition of $\psderiv$.
        Consequently, $E(e)$ covers $\psderiv(e, h_1, h_2)$.

        \item
        If $h_1, h_2 \in \terms_\simeq$ and $\psderiv(e, h_1, h_2) \not\simeq 0$, then $f \simeq h_1$ and $g \simeq h_2$.
        Since $E(f)$ covers $f$ and $E(g)$ covers $g$, $E(e)$ covers both.
    \end{enumerate}

    \item
    If $e = f^*$, we choose $E(e) = E(f) \cup \set{f^*} \cup \set{f' \cdot f^* : f' \in E(f)}$.
    Immediately, $E(e)$ covers $e$.
    For cover-closure of $E(e)$, we need only consider $f' \cdot f^*$ for $f' \in E(f)$.
    \begin{enumerate}[(i)]
        \item
        If $a \in \Sigma$, let $f_0', \dots, f_{n-1}' \in E(f)$ be such that $f' \simeq f_0' + \cdots + f_{n-1}'$.
        We can then derive that
        \begin{align*}
        \ssderiv(f' \cdot f^*, a)
            &= \ssderiv(f', a) \fatsemi f^* + f' \star f^* \\
            &\simeq (f_0' + \cdots + f_{n-1}') \fatsemi f^* + f' \star f^* \\
            &\simeq f_0' \fatsemi f^* + \cdots + f_{n-1}' \fatsemi f^* + f' \star f^*
        \end{align*}
        All of the non-zero terms in the last form can be found in $E(e)$, and thus $E(e)$ covers $\ssderiv(f' \cdot f^*, a)$.

        \item
        If $h_1, h_2 \in \terms_\simeq$, then $\psderiv(f' \cdot f^*, h_1, h_2)$ is covered by $E(e)$ by an argument similar to the above.

        \item
        If $h_1, h_2 \in \terms_\simeq$ and $\psderiv(f' \cdot f^*, h_1, h_2) \not\simeq 0$, then $\psderiv(f', h_1, h_2) \not\simeq 0$.
        Thus $h_1$ and $h_2$ are covered by $E(f)$, and hence by $E(e)$.
    \end{enumerate}

    \item
    If $e = f^\dagger$, we choose $E(e) = E(f) \cup \set{e^\dagger, 1}$.
    Once more, $E(e)$ covers $e$ trivially.
    For cover-closure of $E(e)$, we need only consider $e^\dagger$.
    \begin{enumerate}[(i)]
        \item
        If $a \in \Sigma$, then $\ssderiv(e, a) = 0$, and hence $E(e)$ covers $\ssderiv(e, a)$.

        \item
        If $h_1, h_2 \in \terms_\simeq$, then $\psderiv(e, h_1, h_2) \in \set{0, 1}$ by definition of $\psderiv$, and thus $E(e)$ covers $\psderiv(e, h_1, h_2)$.

        \item
        If $h_1, h_2 \in \terms_\simeq$ and $\psderiv(e, h_1, h_2) \not\simeq 0$, then $h_1 \simeq f$ and $h_2 \simeq e$.
        Since $E(f)$ covers $f$, we conclude that $E(e)$ covers $h_1$ and $h_2$.
        \qedhere
    \end{enumerate}
\end{itemize}
\end{proof}

\subsection{Well-nestedness}

\preceqvsdepth*
\begin{proof}
It suffices to verify the claim for the generating pairs of $\preceq_\Sigma$; in each case, we proceed by induction on $e$.
\begin{itemize}
    \item
    Suppose $e \preceq_\Sigma f$ because $f = \ssderiv(e, a)$ for some $a \in \Sigma$.
    In the base, where $e \in \set{0, 1} \cup \Sigma$, we have that $\ssderiv(e, a) \in \set{0, 1}$, and hence $d_\parallel(\ssderiv(e, a)) = d_\dagger(\ssderiv(e, a)) = 0$ --- the claim then holds immediately.

    For the inductive step, there are five cases to consider.
    Let $\circ \in \set{\parallel, \dagger}$.
    \begin{itemize}
        \item
        If $e = e_0 + e_1$, then $\ssderiv(e, a) = \ssderiv(e_0, a) + \ssderiv(e_1, a)$.
        By induction, we know that $d_\circ(\ssderiv(e_0, a)) \leq d_\circ(e_0)$ and $d_\circ(\ssderiv(e_1, a)) \leq d_\circ(e_1)$.
        We can then derive that
        \begin{align*}
        d_\circ(\ssderiv(e, a))
            &= \max(d_\circ(\ssderiv(e_0, a)), d_\circ(\ssderiv(e_1, a))) \\
            &\leq \max(d_\circ(e_0), d_\circ(e_1))
             = d_\circ(e)
        \end{align*}
        The cases where $e = e_0 \cdot e_1$ or $e = e_0^*$ can be argued similarly.

        \item
        If $e = e_0 \parallel e_1$ or $e = e_0^\dagger$, then $\ssderiv(e, a) = 0$, thus $d_\circ(\ssderiv(e, a)) = 0$.
        The claim then follows immediately.
    \end{itemize}

    \item
    Suppose $e \preceq_\Sigma f$ because $f = \psderiv(e, g, h)$ for some $g, h \in \terms$.
    In the base, where $e \in \set{0, 1} \cup \Sigma$, we have that $\psderiv(e, a) = 0$, and hence $d_\parallel(\ssderiv(e, g, h)) = d_\dagger(\psderiv(e, g, h)) = 0$ --- the claim then holds immediately.

    For the inductive step, there are two cases to consider.
    Let $\circ \in \set{\parallel, \dagger}$.
    \begin{itemize}
        \item
        If $e \in \set{e_0 + e_1, e_0 \cdot e_1, e_0^*}$, then the proof is similar to the corresponding case above.

        \item
        If $e \in e_0 \parallel e_1$ or $e = e_0^\dagger$, then $\psderiv(e, g, h) \in \set{0, 1}$, and hence $d_\circ(\ssderiv(e, a)) = 0$.
        The claim then follows.
    \end{itemize}

    \item
    Suppose $e \preceq_\Sigma f$ because $\psderiv(f, e, h) \not\simeq 0$ or $\psderiv(f, h, e) \not\simeq 0$ for a $h \in \terms$.
    In the base, where $e \in \set{0, 1} \cup \Sigma$, the claim holds vacuously.

    For the inductive step, there are five cases to consider.
    Let $\circ \in \set{\parallel, \dagger}$.
    \begin{itemize}
        \item
        If $e = e_0 + e_1$, then $\psderiv(e, g, h) \not\simeq 0$ implies that $\psderiv(e_0, g, h) \not\simeq 0$ or $\psderiv(e_1, g, h) \not\simeq 0$; w.l.o.g., we assume the former.
        By induction, we have $d_\circ(g), d_\circ(h) \leq d_\circ(e_0)$.
        Since $d_\circ(e_0) \leq d_\circ(e)$, the claim follows.
        The cases where $e = e_0 \cdot e_1$ or $e = e_0^*$ can be argued similarly.

        \item
        If $e = e_0 \parallel e_1$, then $\psderiv(e, g, h) \not\simeq 0$ implies that $e_0 \simeq g$ and $e_1 \simeq h$.
        We also have $d_\parallel(e_0), d_\parallel(e_1) < d_\parallel(e)$, as well as $d_\dagger(e_0), d_\dagger(e_1) \leq d_\dagger(e)$.
        Since $d_\circ(e_0) = d_\circ(g)$ and $d_\circ(e_1) = d_\circ(h)$, the claim follows.

        \item
        If $e = e_0^\dagger$, then $\psderiv(e, g, h) \not\simeq 0$ implies that $g \simeq e_0$ and $h \simeq e$.
        We also have $d_\parallel(e_0) \leq d_\parallel(e)$, as well as $d_\dagger(e_0) < d_\dagger(e)$.
        Since $d_\circ(g) = d_\circ(e_0)$ and $d_\circ(h) = d_\circ(e)$, the claim then follows.
        \qedhere
    \end{itemize}
\end{itemize}
\end{proof}

\daggerloops*
\begin{proof}
  We show by induction on $\preceq_\Sigma$ the following stronger statement: if $e\preceq_\Sigma f$, then the following holds:
  \[
    d_\dagger (e)=d_\dagger (f)\wedge \left(\exists g:f\simeq g^\dagger\right)\Rightarrow e\simeq f.
  \]

  For the base cases, assume $f\simeq g^\dagger$ and $d_\dagger (e)=d_\dagger (f)$.
  This means that $d_\dagger(g) < d_\dagger(g^\dagger) = d_\dagger(e)$.
  \begin{itemize}
  \item if $e\in \left\{\ssderiv(g^\dagger,a),\psderiv(g^\dagger,e_1,e_2)\right\}$, since
    $\ssderiv(g^\dagger,a),\psderiv(g^\dagger,e_1,e_2)\in\left\{0,1\right\}$ it means that $e\in\left\{0,1\right\}$ which in turn implies $d_\dagger(e)=0$.
    This contradicts $d_\dagger(g)<d_\dagger(e)$, therefore the claim holds vacuously.
  \item if $\psderiv(g^\dagger,e,e')\not\simeq 0$, then by definition of $\psderiv$ we know that $e\simeq g$ which contradicts $d_\dagger(g)<d_\dagger(e)$, thus the claim holds vacuously.
  \item if $\psderiv(g^\dagger,e',e)\not\simeq 0$, then by definition of $\psderiv$ we know that $e\simeq g^\dagger$, making the claim hold immediately.
  \end{itemize}
  For the inductive case, assume $e\preceq_\Sigma f_1\preceq_\Sigma f_2$, and assume that $d_\dagger (e)=d_\dagger (f_2)$ and $\exists g: f_2\simeq g^\dagger$.
  By Lemma~\ref{lemma:preceq-vs-depth} we know that $d_\dagger (e)\leqslant d_\dagger (f_1)\leqslant d_\dagger (f_2)=d_\dagger (e)$, meaning that $d_\dagger(e)=d_\dagger(f_1)$ and $d_\dagger(f_1)=d_\dagger(f_2)$.
  Applying the induction hypothesis on the pair $f_1\preceq_\Sigma f_2$ tells us that $f_1\simeq f_2$.
  Since $\exists g:f_2\simeq g^\dagger$, the same holds for $f_1$, so we may apply the induction hypothesis on the pair $e\preceq_\Sigma f_1$ to get $e\simeq f_1$.
  By transitivity we conclude that $e\simeq f_2$.
\end{proof}

\preceqvsforks*
\begin{proof}
  We start by establishing the following statements:
\begin{inparaenum}[(i)]
    \item\label{item:g-smaller}
    $d_\parallel(g) < d_\parallel(e)$ or $d_\dagger(h) < d_\dagger(e)$, as well as
    \item\label{item:h-smaller}
    $d_\parallel(h) < d_\parallel(e)$ or $d_\dagger(h) < d_\dagger(e)$, or $h \simeq f^\dagger$ for some $f \in \terms$.
\end{inparaenum}
The proof for both claims proceeds by induction on $e$.
In the base, where $e \in \set{0, 1} \cup \Sigma$, we have $\psderiv(e, g, h) = 0$, and hence the claim holds vacuously.

For the inductive step, there are three cases to consider.
\begin{itemize}
    \item
    If $e = e_0 + e_1$, then $\psderiv(e_0, g, h) \not\simeq 0$ or $\psderiv(e_1, g, h) \not\simeq 0$; w.l.o.g.\ we assume the former.
    We then have that $d_\parallel(g) < d_\parallel(e_0)$ or $d_\dagger(g) < d_\dagger(e_0)$ by induction; the first claim then follows by definition of $d_\parallel$ and $d_\dagger$.
    We also know that $d_\parallel(h) < d_\parallel(e)$ or $d_\dagger(h) < d_\dagger(e)$ or $h \simeq f^\dagger$ for some $f \in \terms$ by induction.
    In the latter case, the second claim follows immediately; otherwise, the claim follows by definition of $d_\parallel$ and $d_\dagger$ again.

    The cases where $e = e_0 \cdot e_1$ or $e = e_0^*$ can be argued similarly.

    \item
    If $e = e_0 \parallel e_1$, then $g \simeq e_0$ and $h \simeq e_1$ by definition of $\psderiv$.
    We then have $d_\parallel(g) = d_\parallel(e_0) < d_\parallel(e)$, and $d_\parallel(h) < d_\parallel(e)$, satisfying both claims.

    \item
    If $e = e_0^\dagger$, then $g \simeq e_0$ and $h \simeq e$.
    The second claim holds immediately.
    For the first claim, observe that $d_\dagger(g) = d_\dagger(e_0) < d_\dagger(e)$.
  \end{itemize}

  Let us now prove that the statement of the Lemma holds.
  Let $e, g, h \in \terms$ with  $\psderiv(e, g, h) \not\simeq 0$.
  \begin{itemize}
  \item
  By definition of $\preceq_\Sigma$ we know that $g\preceq_\Sigma e$.
  If $e\preceq_\Sigma g$, by Lemma~\ref{lemma:preceq-vs-depth} we would have $d_\dagger(e) \leqslant d_\dagger(g)$.
  However, by~\eqref{item:g-smaller} we have $d_\dagger(g) < d_\dagger(e)$, thus ensuring that $g\prec_\Sigma e$.

  \item We now need to show that either $h\prec_\Sigma e$ or $e\simeq f^\dagger$ for some $f$.
  Since we know that $h\preceq_\Sigma e$, this amounts to showing that if $e\preceq_\Sigma h$ then $e\simeq f^\dagger$.
  Since $h\preceq_\Sigma e\preceq_\Sigma h$ we have $d_\dagger (e)=d_\dagger(h)$ and $d_\parallel (e)=d_\parallel(h)$.
  But according to~\eqref{item:h-smaller} there are three cases: either $d_\parallel(h) < d_\parallel(e)$, or $d_\dagger(h) < d_\dagger(e)$, or $h \simeq f^\dagger$ for some $f \in \terms$.
  The first two cases are in contradiction with what we know so we deduce that $h \simeq f^\dagger$ for some $f \in \terms$.
  Therefore by applying Lemma~\ref{lemma:dagger-loops} to $e\preceq_\Sigma f^\dagger\preceq_\Sigma e$ we get that $e\simeq f^\dagger$.\qedhere
  \end{itemize}
\end{proof}

\subsection{Deconstruction lemmas}%
\label{appendix:deconstruction-lemmas}

\tracedeconstructsequential*
\begin{proof}
The proof proceeds by induction on the length $\ell$ of $e_0 \cdot e_1 \satrace{U} f$.
In the base, where $\ell = 0$, we have that $f = e_0 \cdot e_1$ and $U = 1$.
We can then choose $f_0 = e_0$ and $f_1 = e_1$ as well as $U_0 = U_1 = 1$ to satisfy the claim.

For the inductive step, let $e_0 \cdot e_1 \satrace{U} f$ be of length $\ell+1$.
We find that $U = V \cdot U'$, and a $g \in \terms$ such that $e_0 \cdot e_1 \satrace{V} g$ is a unit trace, and $g \satrace{U} f$ is of length $\ell$.
If $e_0 \cdot e_1 \satrace{V} g$ is a $\delta$-trace, then $V = a$ for some $a \in \Sigma$, and $g = \ssderiv(e_0 \cdot e_1, a) = \ssderiv(e_0, a) \fatsemi e_1 + e_0 \star \ssderiv(e_1, a)$.
By Lemma~\ref{lemma:trace-deconstruct-plus}, we find $f' \in \sacc$ such that $\ssderiv(e_0, a) \fatsemi e_1 \satrace{U'} f'$ or $e_0 \star \ssderiv(e_1, a) \satrace{U'} f'$, of length $\ell$.
This gives us two cases.
\begin{itemize}
    \item
    If $\ssderiv(e_0, a) \fatsemi e_1 \satrace{U'} f'$, then first note that $\ssderiv(e_0, a) \not\simeq 0$, by Lemma~\ref{lemma:trace-top-bottom}, and hence $\ssderiv(e_0, a) \cdot e_1 \satrace{U'} f'$.
    By induction we find $f_0, f_1 \in \sacc$ and $U' = U_0' \cdot U_1'$ such that $\ssderiv(e_0, a) \satrace{U_0'} f_0$ and $e_1 \satrace{U_1'} f_1$, and the total length of these traces is $\ell$.
    We can then choose $U_1 = V \cdot U_0'$ and $U_1 = U_1'$ to find that $U = V \cdot U' = V \cdot U_0' \cdot U_1' = U_0 \cdot U_1$, as well as $e_0 \satrace{U_0} f_0$ and $e_1 \satrace{U_1} f_1$, of total length $\ell+1$.

    \item
    If $e_0 \star \ssderiv(e_1, a) \satrace{U'} f'$, then first note that $e_0 \star \ssderiv(e_1, a) \not\simeq 0$ by Lemma~\ref{lemma:trace-top-bottom}, and so $e_0 \in \sacc$.
    We choose $U_0 = 1$ and $U_1 = U$ as well as $f_0 = e_0$ and $f_1 = f'$ to find that $U = 1 \cdot U = U_0 \cdot U_1$ as well as $e_0 \satrace{U_0} f_0$.
    Furthermore, $e_1 \satrace{V} \ssderiv(e_1, a) = e_0 \star \ssderiv(e_1, a) \satrace{U'} f'$, meaning that $e_1 \satrace{U} f'$.
    The total length of these traces is again $\ell+1$.
\end{itemize}
The case where $e_0 \cdot e_1 \satrace{V} g$ is a $\gamma$-trace can be treated similarly.
\end{proof}

\tracedeconstructstar*
\begin{proof}
The proof proceeds by induction on the length $\ell$ of $e^* \satrace{U} f$.
In the base, where $\ell = 0$, we have that $f = e^*$ and $U = 1$; it suffices to choose $n = 0$.

For the inductive step, let $e^* \satrace{U} f$ be of length $\ell + 1$, and assume that the claim holds for $\ell$.
We then find $g \in \terms$ and $U = V \cdot U'$ such that $e^* \satrace{V} g$ is a unit trace, and $g \satrace{U'} f$ of length $\ell$.
If $e^* \satrace{V} g$ is a $\delta$-trace, then $V = a$ for some $a \in \Sigma$, and $g = \ssderiv(e^*, a) = \ssderiv(e, a) \fatsemi e^*$.
By Lemma~\ref{lemma:trace-top-bottom}, and the fact that $g \satrace{U'} f \in \sacc$, we then know that $\ssderiv(e, a) \not\simeq 0$, and hence $\ssderiv(e, a) \cdot e^* \satrace{U'} f$.
By Lemma~\ref{lemma:trace-deconstruct-sequential}, we find $f'', f' \in \sacc$ such that $U' = W \cdot X$ as well as $\ssderiv(e, a) \satrace{W} f''$ and $e^* \satrace{X} f'$ of total length $\ell$.

Then, by induction, we find $f_1, \dots, f_{n-1} \in \terms$ such that $X = U_1 \cdots U_{n-1}$, and for $1 \leq i < n$ it holds that $e \satrace{X_i} f_i$.
We then choose $f_0 = f''$ and $U_0 = V \cdot W$.
For these choices, $U = U_0 \cdot U' = V \cdot W \cdot X = U_0 \cdots U_{n-1}$.
Since $e \satrace{V} \ssderiv(e, a) \satrace{W} f''$, we also find that $e \satrace{U_0} f_0$; this completes the proof.

The case where $e^* \satrace{V} g$ is a $\gamma$-trace is similar.
\end{proof}

\subsection{Construction lemmas}%
\label{appendix:construction-lemmas}

\traceconstructcarry*
\begin{proof}
The proof proceeds by induction on the length $\ell$ of $e_0 \satrace{U} f_0$.
In the base, where $\ell = 0$, we can choose $f = e_0 \cdot e_1$ to satisfy the claim.

For the inductive step, let $e_0 \satrace{U} f_0$ be of length $\ell + 1$, and assume the claim holds for traces of length $\ell$.
We then find $e_0' \in \terms$ and $U = V \cdot U'$ such that $e_0 \satrace{V} e_0'$ is a unit trace, and $e_0' \satrace{U'} f_0$ is of length $\ell$.
By Lemma~\ref{lemma:trace-top-bottom} and the fact that $e_0' \satrace{U'} f_0 \in \sacc$, we know that $e_0' \not\simeq 0$, and thus $e_0' \fatsemi e_1 = e_0' \cdot e_1$.

By induction, we find $f' \in \terms$ such that $f_0 \cdot e_1 \lesssim f'$, and $e_0' \cdot e_1 \satrace{U'} f'$.
If $e_0 \satrace{V} e_0'$ is a $\delta$-trace, then $V = a$ for some $a \in \Sigma$, and $e_0' = \ssderiv(e_0, a)$.
Then, by Lemma~\ref{lemma:trace-construct-plus}, we find $f \in \terms$ such that $f' \lesssim f$ and $\ssderiv(e_0 \cdot e_1, a) = e_0' \cdot e_1 + e_0 \star \ssderiv(e_1, a) \satrace{U'} f$.
Putting these traces together, we find that $e_0 \cdot e_1 \satrace{U} f$, as well as $f_0 \cdot e_1 \lesssim f' \lesssim f$.

The case where $e_0 \satrace{V} e_0'$ is a $\gamma$-trace is similar.
\end{proof}

\traceconstructledge*
\begin{proof}
The proof proceeds by induction on the length $\ell$ of $e_1 \satrace{V} f_1$.
If $\ell = 0$, we know that $f_1 = e_1$ and $V = 1$.
We can then choose $f = f_0 \cdot e_1$.

For the inductive step, let $e_1 \satrace{V} f_1$ be of length $\ell + 1$, and assume the claim holds for traces of length $\ell$.
We then find $e_1' \in \terms$ and $V = W \cdot V'$ such that $e_1 \satrace{W} e_1'$ is a unit trace, and $e_1 \satrace{V'} f_1$ is of length $\ell$.
If $e_1 \satrace{W} e_1'$ is a $\delta$-trace, then $W = a$ for some $a \in \Sigma$, and $e_1' = \ssderiv(e_1, a)$.
By Lemma~\ref{lemma:trace-construct-plus}, we find $f \in \sacc$ such that $\ssderiv(f_0, a) \fatsemi e_1 + e_1' \satrace{V'} f$.
Since
$
\ssderiv(f_0 \cdot e_1, a)
    = \ssderiv(f_0, a) \fatsemi e_1 + f_0 \star \ssderiv(e_1, a)
    = \ssderiv(f_0, a) \fatsemi e_1 + e_1'
$
we find that $f_0 \cdot e_1 \satrace{W} \ssderiv(f_0, a) \fatsemi e_1 + e_1'$.
We conclude that $f_0 \cdot e_1 \satrace{V} f$.

The case where $e_1 \satrace{W} e_1'$ is a $\gamma$-trace is similar.
\end{proof}

\traceconstructsequential*
\begin{proof}
By Lemma~\ref{lemma:trace-construct-carry}, we find $f' \in \terms$ such that $f_0 \cdot e_1 \lesssim f'$ and $e_0 \cdot e_1 \satrace{U} f'$.
By Lemma~\ref{lemma:trace-construct-ledge}, we find $f'' \in \sacc$ such that $f_0 \cdot e_1 \satrace{V} f''$.
By Lemma~\ref{lemma:trace-construct-plus}, we find $f \in \sacc$ such that $f' \simeq f_0 \cdot e_1 + f' \satrace{V} f$.
In total, we have $e_0 \cdot e_1 \satrace{U \cdot V} f$.
\end{proof}

\traceconstructstar*
\begin{proof}
Without loss of generality, we can assume that for $0 \leq i < n$ it holds that $e \satrace{U_i} f_i$ is non-trivial.
The proof proceeds by induction on $n$.
In the base, where $n = 0$, we can choose $f = e^*$ to satisfy the claim.

For the inductive step, assume that $n > 0$ and that the claim holds for $n-1$.
By induction, we can find $f' \in \sacc$ such that $e^* \satrace{U_1 \cdots U_{n-1}} f'$.
Since $e \satrace{U_1} f_1$ is non-trivial, we find $e' \in \terms$ and $U_0 = V \cdot U_0'$ such that $e \satrace{V} e'$ is a unit trace, and $e' \satrace{U_0'} f_0$.
We note that by Lemma~\ref{lemma:trace-top-bottom}, this implies that $e' \not\simeq 0$.
By Lemma~\ref{lemma:trace-construct-sequential}, we find $f \in \sacc$ such that $e' \cdot e^* \satrace{U_0' \cdot U_1 \cdots U_{n-1}} f$.
If $e \satrace{V} e'$ is a $\delta$-trace, then $V = a$ for some $a \in \Sigma$, and $e' = \ssderiv(e, a)$.
In that case, $e^* \satrace{V} \ssderiv(e^*, a) = \ssderiv(e, a) \fatsemi e^* = e' \cdot e^*$.
Consequently, $e^* \satrace{V} e' \cdot e^* \satrace{U_0' \cdot U_1 \cdots U_{n-1}} f$ and therefore $e^* \satrace{U_0 \cdots U_{n-1}} f$.

The case where $e \satrace{V} e'$ is a $\gamma$-trace is similar.
\end{proof}

\subsection{Soundness of the translation}

\syntacticpalanguages*
\begin{proof}
We proceed by induction on $e$.
In the base, there are two cases to consider.
On the one hand, if $e \in \set{0, 1}$, then the claim follows from Lemma~\ref{lemma:trace-top-bottom}.
On the other hand, if $e = a$ for some $a \in \Sigma$, then the inclusion from left to right is simple: $a \satrace{a} \ssderiv(a, a) = 1 \in \sacc$, and therefore we can conclude that $a \in L_\Sigma(a)$.
For the inclusion from right to left, suppose that $a \satrace{U} f$ for some pomset $U$ and $f \in \sacc$.
In that case, $f \neq a$ (for $a \not\in \sacc$), and thus $e \satrace{U} f$ must be non-trivial.
We therefore find that $U = U_0 \cdot U'$ and a $g \in \terms$ such that $a \satrace{U_0} g$ is a unit trace, and $g \satrace{U'} f$ holds as well.
Whether $a \satrace{U_0} g$ is a $\gamma$-trace or $\delta$-trace, we have that $g \in \set{0, 1}$.
Furthermore, by Lemma~\ref{lemma:trace-top-bottom} and the fact that $f \in \sacc$, we know that $g \satrace{U'} f$ must be trivial (for otherwise $f = 0 \not\in \sacc$), meaning that $g = f = 1$.
It then follows that $a \satrace{U_0} g$ was a $\delta$-trace with $U_0 = a$, and $U = U_0 \cdot U' = a \cdot 1 = a$.

For the inductive step, suppose the claim holds for all strict subterms of $e$.
There are five cases to consider.
\begin{itemize}
    \item
    If $e = e_0 + e_1$, then first suppose that $U \in L_\Sigma(e)$.
    By Lemma~\ref{lemma:trace-deconstruct-plus} we know that $U \in L_\Sigma(e_0)$ or $U \in L_\Sigma(e_0)$.
    By induction, we find that $U \in \sem{e_0} \cup \sem{e_1} = \sem{e_0 + e_1}$.

    For the other inclusion, let $U \in \sem{e_0 + e_1}$.
    If $U \in \sem{e_0}$, then $U \in L_\Sigma(e_0)$ by induction; then $U \in L_\Sigma(e_0 + e_1)$ by Lemma~\ref{lemma:trace-deconstruct-plus}.
    The case where $U \in \sem{e_1}$ is similar.

    \item
    If $e = e_0 \cdot e_1$ (resp.\ $e = e_0^*$), then the equality follows from Lemma~\ref{lemma:trace-deconstruct-sequential} and Lemma~\ref{lemma:trace-construct-sequential} (resp.\ Lemma~\ref{lemma:trace-deconstruct-star} and Lemma~\ref{lemma:trace-construct-star}) by argument analogous to the previous case.

    \item
    If $e = e_0 \parallel e_1$, then first suppose that $U \in L_\Sigma(e_0 \parallel e_1)$.
    A simple look at the sequential and parallel derivatives for $e_0 \parallel e_1$ shows that $U = V \parallel W$ such that $V \in L_\Sigma(e_0)$ and $W \in L_\Sigma(e_1)$.
    By induction, $V \in \sem{e_0}$ and $W \in \sem{e_1}$, and thus $U = V \parallel W \in \sem{e_0 \parallel e_1}$.

    For the other inclusion, suppose that $U \in \sem{e_0 \parallel e_1}$.
    Then $U = V \parallel W$ such that $V \in \sem{e_0}$ and $W \in \sem{e_1}$.
    By induction, we find that $V \in L_\Sigma(e_0)$ and $W \in \sem{e_1}$.
    Another look at the parallel derivatives for $e_0 \parallel e_1$ then tells us that $U = V \parallel W \in L_\Sigma(e_0 \parallel e_1)$.

    \item
    If $e = f^\dagger$, then first note that $f^\dagger$ is a recursive state by Lemma~\ref{lemma:syntactic-pa-well-nested}.
    By Lemma~\ref{lemma:recursive-state-language} and induction, we can then conclude that
    \[
    L_\Sigma(f^\dagger) = {L_\Sigma(f)}^\dagger = \sem{f}^\dagger = \semnostretch{f^\dagger} \qedhere
    \]
\end{itemize}
\end{proof}

\subsection{Soundness modulo congruence}%
\label{appendix:soundness-modulo-congruence}

For technical completeness, we justify our notation in Section~\ref{section:expressions-to-automata} by arguing that the constructs used are well-defined modulo $\simeq$.

\renewcommand{\thesection}{\Alph{section}}

\begin{lemma}%
\label{lemma:congruence-sound}
Let $e, f \in \terms$.
The following hold:
\begin{enumerate}[(i)]
    \item If $e \simeq f$, then $\sem{e} = \sem{f}$, and
    \item if $e \simeq f$, then $e \in \sacc$ if and only if $f \in \sacc$, and
\end{enumerate}
\end{lemma}
\begin{proof}
For the first part, it suffices to show that the claim holds for the pairs generating $\simeq$.
This gives us four cases to consider.
\begin{itemize}
    \item If $e = f + 0$, then $\sem{e} = \sem{f} \cup \sem{0} = \sem{f} \cup \emptyset = \sem{f}$.
    \item If $e = f + f$, then $\sem{e} = \sem{f} \cup \sem{f} = \sem{f}$.
    \item If $e = g_0 + g_1$ and $f = g_1 + g_0$, then $\sem{e} = \sem{g_0} \cup \sem{g_1} = \sem{g_1} \cup \sem{g_0} = \sem{f}$.
    \item If $e = g_0 + (g_1 + g_2)$ and $f = (g_0 + g_1) + g_2$, then
    \[\sem{e} = \sem{g_0} \cup (\sem{g_1} \cup \sem{g_2}) = (\sem{g_0} \cup \sem{g_1}) \cup \sem{g_2} = \sem{f}\]
    \item If $e = (g_0 + g_1) \cdot g_2$ and $f = g_0 \cdot g_2 + g_1 \cdot g_2$, then
    \[\sem{e} = (\sem{g_0} \cup \sem{g_1}) \cdot \sem{g_2} = \sem{g_0} \cdot \sem{g_2} \cup \sem{g_1} \cdot \sem{g_2} = \sem{f}\]
\end{itemize}

\noindent
For the second part, it suffices to verify that the claim holds for the pairs generating $\simeq$.
This gives us again four cases to consider.
\begin{itemize}
    \item Suppose $e = f + 0$.
    If $e \in \sacc$, then either $f \in \sacc$ or $0 \in \sacc$.
    Since the latter is false, $f \in \sacc$.
    Also, if $f \in \sacc$, then $e = f + 0 \in \sacc$ immediately.
    \item Suppose $e = f + f$.
    If $f + f \in \sacc$, then $f \in \sacc$; if $f \in \sacc$, then $f + f \in \sacc$.
    \item Suppose $e = g_0 + g_1$ and $f = g_1 + g_0$.
    If $g_0 + g_1 \in \sacc$, then $g_0 \in \sacc$ or $g_1 \in \sacc$; in either case, $g_1 + g_0 \in \sacc$.
    The other direction is analogous.
    \item Suppose $e = g_0 + (g_1 + g_2)$ and $f = (g_0 + g_1) + g_2$.
    If $e \in \sacc$, then $g_0 \in \sacc$ or $g_1 + g_2 \in \sacc$, and thus one of $g_0, g_1, g_2$ must be in $\sacc$.
    But then $g_0 + g_1$ or $g_2$ must be in $\sacc$, and thus $f = (g_0 + g_1) + g_2 \in \sacc$.
    The proof in the other direction is similar.
    \item Suppose $e = (g_0 + g_1) \cdot g_2$ and $f = g_0 \cdot g_2 + g_1 \cdot g_2$.
    If $e \in \sacc$, then $g_0 + g_1 \in \sacc$ and $g_2 \in \sacc$, meaning that $g_0$ or $g_1$ can be found in $\sacc$, and $g_2$ too.
    In that case, either $g_0$ and $g_2$, or $g_1$ and $g_2$ can be found in $\sacc$, and thus $f \in \sacc$.
    The proof in the other direction is similar.
    \qedhere
\end{itemize}
\end{proof}

\begin{lemma}%
\label{lemma:congruent-states-derive}
Let $e, f \in \terms$ such that $e \simeq f$.
The following hold:
\begin{enumerate}[(i)]
    \item If $a \in \Sigma$, then $\ssderiv(e, a) \simeq \ssderiv(f, a)$.
    \item If $g, h, g', h' \in \terms$ with $g \simeq g'$ and $h \simeq h'$, then $\psderiv(e, g, h) = \psderiv(f, g', h')$.
\end{enumerate}
\end{lemma}
\begin{proof}
For the first part, it suffices to verify the claim for the pairs generating $\simeq$.
This gives us four cases to consider.
\begin{itemize}
    \item If $e = f + 0$, then $\ssderiv(e, a) = \ssderiv(f, a) + \ssderiv(0, a) = \ssderiv(f, a) + 0 \simeq \ssderiv(f, a)$.
    \item If $e = f + f$, then $\ssderiv(e, a) = \ssderiv(f, a) + \ssderiv(f, a) \simeq \ssderiv(f, a)$.
    \item If $e = g_0 + g_1$ and $f = g_1 + g_0$, then
    \begin{align*}
    \ssderiv(e, a)
        &= \ssderiv(g_0, a) + \ssderiv(g_1, a) \\
        &\simeq \ssderiv(g_1, a) + \ssderiv(g_0, a)
         = \ssderiv(f, a)
    \intertext{%
        \item If $e = g_0 + (g_1 + g_2)$ and $f = (g_0 + g_1) + g_2$, then
    }
    \ssderiv(e, a)
        &= \ssderiv(g_0, a) + (\ssderiv(g_1, a) + \ssderiv(g_2, a)) \\
        &\simeq (\ssderiv(g_0, a) + \ssderiv(g_1, a)) + \ssderiv(g_2, a) \\
        &= \ssderiv(f, a)
    \end{align*}
    \item
    If $e = (g_0 + g_1) \cdot g_2$ and $f = g_0 \cdot g_2 + g_1 \cdot g_2$, then
    \begin{align*}
    \ssderiv(e, a)
        &= \ssderiv(g_0 + g_1, a) \fatsemi g_2 + (g_0 + g_1) \star \ssderiv(g_2, a) \\
        &= (\ssderiv(g_0, a) + \ssderiv(g_1, a)) \fatsemi g_2 + (g_0 + g_1) \star \ssderiv(g_2, a) \\
        &\simeq \ssderiv(g_0, a) \fatsemi g_2 + \ssderiv(g_1, a) \fatsemi g_2 + g_0 \star \ssderiv(g_2, a) + g_1 \star \ssderiv(g_2, a) \\
        &\simeq \ssderiv(g_0 \cdot g_2 + g_1 \cdot g_2, a)
         = \ssderiv(f, a)
    \end{align*}
    in which we make use of the fact that $e + f \simeq 0$ if and only if $e \simeq 0$ and $f \simeq 0$.
    The implication from right to left follows from $e + f \simeq 0 + 0 \simeq 0$, and the other implication from the fact that $e \simeq e + 0 \simeq e + e + f \simeq e + f \simeq 0$, and similarly for $f$.
\end{itemize}

\noindent
For the second part, note that $\psderiv(f, g, h) \simeq \psderiv(f, g', h')$ by construction of $\psderiv$.
It therefore suffices to verify that $\psderiv(e, g, h) \simeq \psderiv(f, g, h)$ for the pairs generating $\simeq$.
This gives us four cases to consider, all of which go through in the same manner as above.
\end{proof}

\begin{lemma}%
\label{lemma:depth-vs-simeq}
Let $e \simeq f$.
Then $d_\parallel(e) = d_\parallel(f)$ and $d_\dagger(e) = d_\dagger(f)$.
\end{lemma}
\begin{proof}
Let $\circ \in \{ \parallel, \dagger \}$.
It suffices to verify the claim for the generating pairs.
\begin{itemize}
    \item
    If $e = f + 0$, then $d_\circ(e) = \max(d_\circ(f), d_\circ(0)) = d_\circ(f)$.

    \item
    If $e = f + f$, then $d_\circ(e) = \max(d_\circ(f), d_\circ(f)) = d_\circ(f)$.

    \item
    If $e = e_0 + e_1$ and $f = e_1 + e_0$, then
    \begin{align*}
    d_\circ(e)
        &= \max(d_\circ(e_0), d_\circ(e_1)) \\
        &= \max(d_\circ(e_1), d_\circ(e_0))
         = d_\circ(f)
    \intertext{%
        \item
        If $e = e_0 + (e_1 + e_2)$ and $f = (e_0 + e_1) + e_2$, then
    }
    d_\circ(e)
        &= \max(d_\circ(e_0), \max(d_\circ(e_1), d_\circ(e_2))) \\
        &= \max(\max(d_\circ(e_0), d_\circ(e_1)), d_\circ(e_2))
         = d_\circ(f)
    \intertext{%
        \item
        If $e = e_0 \cdot (e_1 + e_2)$ and $f = e_0 \cdot e_1 + e_0 \cdot e_2$, then
    }
    d_\circ(e)
        &= \max(d_\circ(e_0), \max(d_\circ(e_1), d_\circ(e_2))) \\
        &= \max(\max(d_\circ(e_0), d_\circ(e_1)), \max(d_\circ(e_0), d_\circ(e_2)))
         = d_\circ(f)
        \tag*{\qedhere}
    \end{align*}
\end{itemize}
\end{proof}

\end{document}